\documentclass[letterpaper,11pt]{article} 

\usepackage{amsmath,amsthm}
\usepackage{amssymb}

\usepackage[comma,longnamesfirst]{natbib}
\usepackage[dvips]{epsfig}
\usepackage{dcolumn}
\usepackage{enumerate}
\usepackage{hhline}
\usepackage{dsfont}
\usepackage{afterpage}
\usepackage{arydshln}
\usepackage{graphicx}
\usepackage{color}
\usepackage[usenames,dvipsnames]{xcolor}
\usepackage{rotating}
\usepackage{hyperref}
\usepackage[percent]{overpic}
\usepackage[]{caption}
\usepackage{algorithmicx}
\usepackage[]{algpseudocode}
\usepackage{algorithm}
\usepackage{diagbox}
\usepackage{graphicx}
\usepackage{wrapfig}
\usepackage{lscape}
\usepackage{fontenc}
\usepackage{setspace}
\usepackage{pdfpages}

\usepackage{bm}
\usepackage{slashbox}
\usepackage{lscape}
\usepackage{breakurl} 
\usepackage{multirow,booktabs}
\usepackage{eurosym}
\usepackage{titlesec}
\usepackage{sectsty}
\usepackage{placeins}
\usepackage{subcaption}
\usepackage[flushleft]{threeparttable}
\usepackage{appendix}

\epsfverbosetrue
\def\theequation{\thesection.\arabic{equation}}  
\def\abstract{\if@twocolumn
\section*{Abstract}
\else \normalsize 
\begin{center}
{\bf Summary\vspace{-.5em}\vspace{0pt}} 
\end{center}
\quotation 
\fi}
\def\endabstract{\if@twocolumn\else\endquotation\fi}

\makeatletter
\newcommand{\myappendix}[1]{
	\setcounter{section}{1}
        \renewcommand{\thesection}{A\arabic{section}}}

\DeclareMathOperator{\E}{E}

\DeclareMathOperator*{\argmin}{{arg\,min}}
\DeclareMathOperator*{\argmax}{{arg\,max}}

\def \calD {\mathcal D}

\def \calF {\mathcal F}

\def \calL {\mathcal L}

\def \uvec {\text{\boldmath$u$}}    
\def \vvec {\text{\boldmath$v$}}    
    
\def \xvec {\text{\boldmath$x$}}    
\def \yvec {\text{\boldmath$y$}}

\def \etavec          {\text{\boldmath$\eta$}}
\def \thetavec        {\text{\boldmath$\theta$}}

\def \lambdavec       {\text{\boldmath$\lambda$}}
\def \muvec           {\text{\boldmath$\mu$}}

\def \phivec          {\text{\boldmath$\phi$}}
\def \varphivec       {\text{\boldmath$\varphi$}}
\def \psivec          {\text{\boldmath$\psi$}}

\usepackage{color}
\usepackage{colordvi}
\newlength{\breite}
\breite\textwidth
\newcounter{aufg}[section]
  {\refstepcounter{aufg}\noindent\textbf{Exercise \arabic{aufg}:}
   \\*[1ex]\noindent}{\vspace{.5cm}}
   
 \newcounter{notes}[section]
  {\refstepcounter{aufg}\noindent\textbf{}
   \\*[1ex]\noindent}{\vspace{.5cm}}
   
\usepackage{amsthm}  
\newtheorem{defin}{Definition} 

\newtheorem{lemma}[defin]{Lemma}

\theoremstyle{definition}

\newtheorem*{beisp*}{Example}
\newtheorem{Proof}{Proof}


\newtheoremstyle{break}
  {}
  {}
  {}
  {}
  {\bfseries}
  {.}
  {\newline}
  {}
  
\theoremstyle{break}

\newcommand{\head}[2]%
 {\hrule \vspace{.15cm} {\sfbold Advanced Statistical Inference, Summer Term 2012, Georg-August-University G\"ottingen}\hfill
{\sfbold Sheet #1}\\
{\sfbold Prof. Dr. Thomas Kneib, Nadja Klein}\hfill {\sfbold #2}

\vspace{.2cm}
\hrule

\vspace{1cm}

}

\newcounter{auf}
{\refstepcounter{auf}
\begin{center}
\fcolorbox[gray]{0}{.95}{
\makebox[\breite]{
\textbf{Exercise \arabic{auf}}
}}\\*[1ex]\noindent
\end{center}
}{\vspace{.5cm}}

\newcounter{loes}[section]
{\stepcounter{loes}
\begin{center}
\fcolorbox[gray]{0}{.95}{
\makebox[\breite]{
\textbf{L"osung \arabic{loes}}
}}\\*[1ex]\noindent
\end{center}
}{}

{\begin{center}
\fcolorbox[gray]{0}{.95}{
\makebox[\breite]{
\textbf{Zu Aufgabe #1}
}}\\*[1ex]\noindent
\end{center}\vspace{1cm}
}{\vspace{1cm}}

\newcounter{ka}
{\refstepcounter{ka}
\begin{center}
\framebox[\textwidth]{
\textbf{Aufgabe \arabic{ka}} \hfill #1 Punkte
}\\*[1ex]\noindent
\end{center}
}{\vspace{1cm}}

\newcounter{lka}
{\refstepcounter{lka}
\begin{center}
\framebox[\textwidth]{
\textbf{L\"osung \arabic{lka}} \hfill #1 Punkte
}\\*[1ex]\noindent
\end{center}
}{\vspace{1cm}}

\usepackage[margin=1in]{geometry}
\titlespacing*\section{0pt}{0pt plus 4pt minus 2pt}{0pt plus 2pt minus 2pt}
\titlespacing*\subsection{0pt}{0pt plus 4pt minus 2pt}{0pt plus 2pt minus 2pt}
\titlespacing*\subsubsection{0pt}{0pt plus 4pt minus 2pt}{0pt plus 2pt minus 2pt}

\definecolor{myblue}{RGB}{0,73,114}
\allsectionsfont{\sffamily\color{myblue}}

\newcounter{myremark}

\newcounter{mynotation}

\usepackage{paralist}

\renewenvironment{itemize}[1]{\begin{compactitem}#1}{\end{compactitem}}
\renewenvironment{enumerate}[1]{\begin{compactenum}#1}{\end{compactenum}}

\newcommand{\mycomment}[1]{}
\newcommand{\bt}{\bm{\theta}}
\newcommand{\be}{\bm{\eta}}
\newcommand{\bp}{\bm{\psi}}

\newcommand{\MD}{\mathcal{D}}
\newcommand{\cut}{\mathrm{cut}}

\theoremstyle{plain}
\newtheorem{theorem}{Theorem}
\newtheorem{corollary}{Corollary}[theorem]
\newtheorem{assumption}{Assumption}

\makeatletter
\def\@seccntformat#1{\@ifundefined{#1@cntformat}%
	{\csname the#1\endcsname\quad}  
	{\csname #1@cntformat\endcsname}
}
\let\oldappendix\appendix 
\renewcommand\appendix{%
	\oldappendix
	\newcommand{\section@cntformat}{\appendixname~\thesection\quad}
}
\makeatother

\usepackage{titlesec}

\usepackage{scalerel,stackengine}
\stackMath
\newcommand\reallywidehat[1]{%
\savestack{\tmpbox}{\stretchto{%
  \scaleto{%
    \scalerel*[\widthof{\ensuremath{#1}}]{\kern-.6pt\bigwedge\kern-.6pt}%
    {\rule[-\textheight/2]{1ex}{\textheight}}
  }{\textheight}%
}{0.5ex}}%
\stackon[1pt]{#1}{\tmpbox}%
}

\renewcommand{\baselinestretch}{1.8}

\begin{document}
\setlength{\abovedisplayskip}{0.15cm}
\setlength{\belowdisplayskip}{0.15cm}
\pagestyle{empty}
\begin{titlepage}

\title{\bfseries\sffamily\color{myblue}  
	 Cutting Feedback in Misspecified Copula Models}
\author{Michael Stanley Smith, Weichang Yu, David J. Nott and David T. Frazier}
\date{\today}
\maketitle
\noindent
{\small Michael Smith is Professor of Management (Econometrics) at Melbourne Business School in the University of Melbourne, Weichang Yu is Lecturer in the School of Mathematics and Statistics at the University of Melbourne, David Nott is an Associate Professor in the Department of Statistics and Data Science at the National University of Singapore and is affiliated with the NUS Institute of Operations Research and Analytics, and 
David Frazier is Professor in Econometrics and Business Statistics at Monash University. David Frazier acknowledges funding from the Australian Research Council under Projects DE200101070 and DP200101414, while
David Nott and Michael Smith's research was supported by the Ministry of Education, Singapore, under the Academic Research Fund Tier 2 (MOE-T2EP20123-0009). The authors thank two anonymous referees whose insightful comments have improved the paper.
}

\newpage
\begin{center}
\mbox{}\vspace{2cm}\\
{\LARGE \title{\bfseries\sffamily\color{myblue} Cutting Feedback in Misspecified Copula Models}
}\\
\vspace{1cm}
{\Large Abstract}
\end{center}
\vspace{-1pt}
\onehalfspacing
\noindent
In copula models the marginal distributions and 
copula function are specified separately. We treat these as two modules in a modular Bayesian inference framework, and propose conducting modified Bayesian inference by ``cutting feedback''. 
Cutting feedback 
limits the influence of potentially misspecified modules in posterior inference. 
We consider two types of cuts. The first limits
the influence of a misspecified copula on inference for the marginals, which is 
a Bayesian analogue of the popular Inference for Margins (IFM) estimator. The second limits the
influence of misspecified marginals on inference for the copula parameters by using
a pseudo likelihood of the ranks to define the cut model.  We establish that if only one of the modules is misspecified, then the appropriate cut posterior gives accurate uncertainty quantification asymptotically for the parameters in the other module. Computation of the cut posteriors
is difficult, and new variational inference methods to do so are proposed.
The efficacy of the new methodology is demonstrated using both simulated data and a substantive multivariate time series
application from macroeconomic forecasting. 
In the latter, cutting feedback from misspecified marginals to a 1096 dimension copula
improves posterior inference and predictive accuracy greatly, compared
to conventional Bayesian inference.

\vspace{20pt}
 
\noindent
{\bf Keywords}:  Modular Inference, Posterior Consistency, Pseudo Rank Likelihood, Time Series Copula, Variational Inference. 

\end{titlepage}

\newpage
\pagestyle{plain}
\setcounter{equation}{0}
\renewcommand{\theequation}{\arabic{equation}}

\section{Introduction}\label{sec:intro}
A copula model 
specifies a multivariate distribution using its marginal distributions and a copula function to capture the dependence structure~\citep{nelsen06,joe2014dependence}. This  simplifies multivariate stochastic modelling, making copula models popular in hydrology~\citep{genest2007metaelliptical},
financial econometrics~\citep{patton2006},
transportation studies~\citep{bhat2009} 
and elsewhere.
However, exploiting this modularity of copula models to improve the accuracy of statistical inference has been less explored.
The goal of the present work is to do so using Bayesian modular inference methods which, to the best of our knowledge, 
have not been considered previously for copula models. We use a technique called ``cutting feedback'' \citep{liu+bb09,jacob+mhr17}, which is applicable to models that comprise multiple components labelled
modules. If some modules
are misspecified, cutting feedback modifies conventional Bayesian inference to limit the influence of the unreliable modules on the others, producing a ``cut posterior''
that is more accurate than a conventional posterior for this case. By treating the marginals of a copula model as one module, and the copula function as a second module, we specify two types of
cut posterior and develop methods for their evaluation. We establish both theoretically and empirically that the cut posteriors are more accurate than the conventional posterior
under the given module misspecification. 

Conventional Bayesian inference for copula models 
using the joint posterior can be
unreliable when either the copula function or the marginals are misspecified, and we consider both scenarios. 
In the first scenario, the goal is to
prevent misspecification of the copula fuction with unknown parameters $\psivec$ from contaminating inference
about the marginals with unknown parameters $\thetavec$. 
We construct a 
joint cut posterior for $(\thetavec,\psivec)$ by cutting feedback
from $\psivec$ to $\thetavec$. 
The approach is a Bayesian analogue
of the IFM method~\citep{joexu1996,joe2005}, but with Bayesian propagation of uncertainty.  
We prove that this cut posterior
asymptotically quantifies uncertainty
about $\thetavec$ accurately 
when the marginal distributions
are well-specified, even if the copula function is misspecified, and 
that the cut posterior mean is first-order equivalent
to the IFM estimator.  


The second scenario is where the goal is to
prevent misspecification of the marginals from contaminating inference
about the copula function. This is the more challenging case.  
To cut feedback from $\thetavec$ to $\psivec$ we define a novel marginal\footnote{The term ``marginal'' has two usages here. The first is to refer to the marginal distributions of the copula model, while the second is to Bayesian marginal cut or conventional posterior distributions. Care is taken to clarify between the two throughout.} cut posterior of $\psivec$ using a pseudo likelihood of the rank data. This is then combined with the conventional conditional posterior for 
$\thetavec$, given $\psivec$, to define
a joint cut posterior.   
We prove that this cut posterior
asymptotically 
quantifies uncertainty accurately about the copula parameters
if the copula is correctly specified, even if the
marginal distributions are misspecified. We are unaware of any existing analog to this approach.  

Computation of cut posteriors using Markov chain Monte Carlo (MCMC) or importance sampling methods is difficult because a cut posterior features an
unignorable intractable term. Nested MCMC methods 
are a common approach
to circumvent this problem~\citep{plummer15}, but they do not scale well. A major contribution of this paper is the development of efficient and scalable variational methods for the evaluation of the cut posteriors outlined above. 
Variational inference \citep{ormerod2010,blei2017} formulates the
approximation of a Bayesian posterior
distribution as an optimization problem. It is particularly attractive for evaluating cut posterior distributions because the problematic intractable term does not need to be computed during the optimization~\citep{yu+ns21,carmona+n22}.
We show how to implement variational inference for both 
cut posteriors of the copula model. 

In the scenario where feedback from $\thetavec$ to $\psivec$ is cut, the introduction of the pseudo likelihood of the ranks adds an additional computational bottleneck because it is difficult to evaluate or optimize directly in even moderate dimensions. To solve this problem we introduce
an extended likelihood~\citep{PitChaKoh2006,hoff07,smith2012estimation} and then define a cut version of the resulting augmented 
posterior which is both tractable and has the desired 
cut posterior as its marginal in $(\thetavec,\psivec)$. We develop variational
approximations to this augmented cut posterior that are both accurate
and allow for fast solution of the variational optimization.

We demonstrate the efficacy of the proposed methodology 
in the presence of both forms of module misspecification using two simulation studies. These are low dimensional to allow evaluation of the exact cut posteriors, which is difficult otherwise. Here, the cut posteriors provide substantially improved statistical inference in comparison to the conventional posterior, while their variational approximations are also shown to be accurate. 
The effectiveness of the variational inference methodology is demonstrated in a substantive macroeconomic application where it is infeasible to use exact methods. The example updates the four-dimensional multivariate time series analysis of~\cite{smithvahey2016} to contemporary data. 
A copula model is used with four unique skew-t marginals and a Gaussian copula of dimension 1096.
This copula arises from a four-dimensional Gaussian copula process with 72 unique parameters observed at 274 time points. The primary objective is density and tail forecasting, and parametric marginals and copula function are necessary to do so, but both are difficult to select. We show that cutting feedback from the marginals to the copula (i.e. from $\thetavec$ to $\psivec$) improves statistical inference and forecasting accuracy substantially, compared to the conventional posterior. 

The rest of the paper is structured as follows.  Section~\ref{sec:bcm} gives some necessary background on copula models, variational inference and cutting feedback. Section~\ref{sec:cutting1} considers 
cutting feedback when the copula is misspecified, but the marginals are not. Both the theoretical behaviour and computation of the cut posterior are discussed and demonstrated in a simulation study. 
Section~\ref{sec:cutting2} considers cutting feedback when the marginals are misspecified, but the copula is not. The theoretical behaviour of the cut posterior is established, while an augmented posterior and an appropriate variational approximation is proposed for its computation, and a simulation study demonstrates. Section~\ref{sec:macro} 
contains the macroeconomic application, while 
Section~\ref{sec:conc} concludes.
\section{Background}\label{sec:bcm}
We begin with a brief outline of copula models and their 
estimation using the conventional joint posterior. This is followed by an introduction to cutting 
feedback methods and the computation of cut posteriors for a model with two modules.
\subsection{Copula Models}
If $\bm{Y}=(Y_1,\ldots,Y_m)^\top\sim F_Y$ with marginals $Y_j\sim F_{j}$,
then the joint distribution function
\begin{equation}
	F_Y(\yvec)=C\left(F_1(y_1),\ldots,F_m(y_m)\right)\,,\label{eq:copmod}
\end{equation}
where $\yvec=(y_1,\ldots,y_m)^\top$ and $C:[0,1]^m \rightarrow \mathbb{R}^+$ is 
a copula function; see~\citet[p.45]{nelsen06}. 
This decomposition provides a convenient modular way to construct a multivariate distribution, where the marginals $F_1,\ldots,F_m$ and copula function $C$ can be selected separately.  Parametric marginals 
$F_j(y_j;\thetavec_j)$ and copula function 
$C(\uvec;\psivec)$ are often used,\footnote{The notation $C(\uvec;\psivec)$ and $C(u_1,\ldots,u_m;\psivec)$ are used interchangably throughout the paper, as are $c(\uvec;\psivec)$ and $c(u_1,\ldots,u_m;\psivec)$ for the copula density.} with $\uvec=(u_1,\ldots,u_m)^\top$ and parameters $\thetavec=(\thetavec_1^\top,\ldots,
\thetavec_m^\top)^\top$ and $\psivec$. The resulting distribution $F_Y$ is commonly called a ``copula model'' and is employed widely.
Many copula
functions with different dependence properties have been studied previously; 
see~\cite{nelsen06} and \cite{joe2014dependence} for some examples.

If $F_1,\ldots,F_m$ are all continuous distributions, the joint density of $\bm{Y}$ is 
\begin{equation}
	f_Y(\yvec;\thetavec,\psivec)=c\left(F_1(y_1;\thetavec_1),\ldots,F_m(y_m;\thetavec_m);\psivec\right)\prod_{j=1}^m f_j(y_j;\thetavec_j)\,,\label{eq:copden}
\end{equation}
where $c(\uvec;\psivec)=\frac{\partial}{\partial\uvec}C(\uvec;\psivec)$ is called the copula density, and $f_j(y_j;\thetavec_j)=\frac{\partial}{\partial y_j} F_j(y_j;\thetavec_j)$ is the marginal density
of $Y_j$. When one or more $F_j$ is discrete or mixed, the joint mixed density function involves differencing over those dimensions;
see~\cite{genest2007}.

Let ${\cal D}=\{\yvec_1,\ldots,\yvec_n\}$ be $n$ observations drawn independently from $F_Y$
at~\eqref{eq:copmod}. For continuous marginals,
the joint posterior density is $p(\thetavec,\psivec|{\cal D})\propto
\prod_{i=1}^n f_Y(\yvec_i|\thetavec,\psivec)p(\thetavec,\psivec)$, where 
 $p(\thetavec,\psivec)$ is the prior. 
Evaluation of the posterior 
 using Markov chain Monte Carlo (MCMC) methods
has been discussed previously by~\cite{PitChaKoh2006,silva2008copula,min2010,smithmin2010} and~\cite{murray2013} among others. 
However, evaluation of the posterior using MCMC methods can be slow for large $m$, and variational inference (VI) is a faster and more scalable alternative. 

\subsection{Variational Inference}
MCMC methods evaluate the  posterior exactly (up to a controllable level of Monte Carlo error),
whereas VI approximates the posterior by a density chosen from a family
of tractable distributions with densities $q\in\calF$. The density is chosen to minimize 
the distance between the two, with the Kullback-Leibler (KL) divergence
the most commonly used measure, so that for a copula model
\begin{equation}
	q^*(\thetavec, \psivec) = \argmin_{q \in \calF} \int \int q(\thetavec, \psivec) \log \left \{ \frac{q(\thetavec, \psivec)}{p(\thetavec, \psivec \, \vert \calD)} \right \} \, \text{d} \thetavec d \psivec.\label{eq:viprob}
\end{equation}
Many families $\calF$ have been considered in the literature, but a Gaussian with 
density $q_\lambda(\xvec)=\phi_N(\xvec;\muvec,\Sigma)$ indexed by its unique parameters $\lambdavec=(\muvec^\top,\mbox{vech}(\Sigma)^\top)^\top$ is one of the most popular~\citep{titsias2014doubly,kucukelbir2017automatic,tan2018gaussian}. 
It is straightforward to show (e.g. see~\cite{ormerod2010}) that $q^*(\thetavec, \psivec) = \argmax_{q \in \calF}\calL(\lambdavec)$, where
the function
\[
\calL(\lambdavec)=E_q\left(\log h(\thetavec,\psivec)-\log q_\lambda(\thetavec,\psivec)\right)\,,  
\] 
is called the Evidence Lower Bound (ELBO) and
$h(\thetavec,\psivec)=p(\calD|\thetavec,\psivec)p(\thetavec,\psivec)$. 

A popular way to solve
this problem is to use stochastic gradient
optimization~\citep{bottou10}. 
This employs an unbiased approximation of the gradient $\nabla_\lambda \calL(\lambdavec)$ along with automatic adaptive step sizes for
the updates of $\lambdavec$, such as the ADADELTA method of~\cite{zeiler12} that
we use here. The combination of stochastic optimization and generic approximations is often called black box VI~\citep{ranganath14,titsias2014doubly}.
VI has been used to estimate copula models 
by~\cite{loaiza2019VBDA}, \cite{nguyen2020VI} and~\cite{smithklein2021}.

\subsection{Cutting feedback methods}\label{sec:cfm}
Cutting feedback is a form of Bayesian
modular inference~\citep{liu+bb09} that removes
the impact of mis-specifying one or more model components
on inference for the other components. 
Comprehensive overviews of cutting feedback methods are provided by
\cite{lunn+bsgn09}, \cite{plummer15}, \cite{jacob+mhr17} and~\cite{yu+ns21}.
A short introduction is 
given here for a two module 
system because 
the methods developed later for copula models are two module systems.

Consider a model for data ${\cal D}$ with density $g({\cal D}|\bm{\eta})$ and parameter vector $\bm{\eta}$. Consider the partition
$\bm{\eta}=(\etavec_1^\top,\etavec_2^\top)^\top$
and assume the density can be factorized as
\begin{equation}
	g({\cal D}|\bm{\eta}) = g_1({\cal D}|\etavec_1)g_2({\cal D}|\etavec_1,\etavec_2)\,.\label{eq:cflikefactor}
\end{equation}
Often in a two module system the data consists
of two sources ${\cal D}_1$ and ${\cal D}_2$, with $g_1({\cal D}|\etavec_1)=g_1({\cal D}_1|\etavec_1)$ and
$g_2({\cal D}|\etavec_1,\etavec_2)=g_2({\cal D}_2|\etavec_1,\etavec_2)$; e.g. 
see~\cite{plummer15}. 
However, we do not assume this simplification here because a more general perspective,
where $g_1$ and $g_2$ represent
different terms in a decomposition of the likelihood, is needed for the
cut methods for copulas. 

Denoting the prior density as $p(\etavec)=p(\etavec_1)p(\etavec_2|\etavec_1)$, 
we define two ``modules'', 
with Module~1 consisting of
$g_1({\cal D}|\etavec_1)$ and $p(\etavec_1)$, and Module~2
consisting of $g_2({\cal D}|\etavec_1,\etavec_2$) and $p(\etavec_2|\etavec_1)$. 
The conventional joint posterior density is
$$p(\etavec_1,\etavec_2|{\cal D})=
p(\etavec_1|{\cal D})\times p(\etavec_2|\etavec_1,{\cal D}).$$
Writing
$\bar{g}({\cal D})=\int p(\etavec_1)p(\etavec_2|\etavec_1) g({\cal D}|\bm{\eta})d\bm{\eta},$
and 
$\bar{g}_2({\cal D}|\etavec_1)=\int p(\etavec_2|\etavec_1)g_2({\cal D}|\etavec_1,\etavec_2)d\etavec_2,$
a simple derivation shows that the marginal posterior density is
\begin{align}
p(\etavec_1|{\cal D}) & =\frac{p(\etavec_1)g_1({\cal D}|\etavec_1)\bar{g}_2({\cal D}|\etavec_1)}{\bar{g}({\cal D})},  \label{theta-marginal}
\end{align}
and the conditional posterior density is
\begin{align}
 p(\etavec_2|\etavec_1,{\cal D}) & = 
 \frac{p(\etavec_2|\etavec_1)g_2({\cal D}|\etavec_1,\etavec_2)}{\bar{g}_2({\cal D}|\etavec_1)}. \label{psi-conditional}
\end{align}
In~\eqref{theta-marginal}, $\bar{g}_2({\cal D}|\etavec_1)$ is 
called the ``feedback'' term, because it captures the effect
of Module~2 on inference for $\etavec_1$. 
If Module~2 is misspecified 
 the
influence of the feedback term can result in misleading marginal
inference for $\etavec_1$.  Hence in joint Bayesian inference, 
even if Module~1 is correctly specified, misspecification of Module~2
can result in misleading inference about parameters appearing
in both modules.  

To eliminate the impact of a misspecification of Module~2 on 
inference for $\etavec_1$, the feedback term can be removed from~\eqref{theta-marginal} to define the following
marginal cut posterior density
\begin{align}
  p_{\text{cut}}(\etavec_1|{\cal D}) & = \frac{p(\etavec_1)g_1({\cal D}|\etavec_1)}{\int p(\etavec_1')g_1({\cal D}|\etavec_1')\,d\etavec_1'}. \label{cut-theta-marginal}
\end{align} 
The joint cut posterior density is then defined as
\begin{align}
     p_{\text{cut}}(\etavec_1,\etavec_2|{\cal D}) & = 
     p_{\text{cut}}(\etavec_1|{\cal D})p(\etavec_2|\etavec_1,{\cal D}), \label{cut-joint}
\end{align}
where $p(\etavec_2|\etavec_1,{\cal D})$ is the same conditional posterior for the
cut and uncut cases. A key observation is that uncertainty about $\etavec_1$ is still propagated when computing marginal cut posterior inference for $\etavec_2$ with
\[p_{\text{cut}}(\etavec_2|{\cal D})=\int p_{\text{cut}}(\etavec_1,\etavec_2|{\cal D})  d\etavec_1\,.
\]  

\subsection{Exact cut posterior computation}
Cut posterior computation is difficult.    
The joint cut posterior density is
$$p_{\text{cut}}(\etavec_1,\etavec_2|{\cal D})
\propto \frac{p(\etavec_1)p(\etavec_2|\etavec_1)g_1({\cal D}|\etavec_1)g_2({\cal D}|\etavec_1,\etavec_2)}{\bar{g}_2({\cal D}|\etavec_1)},$$
where $\bar{g}_2({\cal D}|\etavec_1)$ is 
usually intractable. This makes it hard to implement 
MCMC or importance sampling methods to evaluate the cut posterior in many 
models. 
One approach is to draw samples from~\eqref{cut-joint} by first
drawing $\etavec_1'\sim p_{\text{cut}}(\etavec_1|{\cal D})$, and then $\etavec_2'|\etavec_1'\sim p(\etavec_2|\etavec_1',{\cal D})$. Because $\etavec_1'$ is fixed in the second stage, the intractable term
$\bar{g}_2({\cal D}|\etavec_1')$ is not computed. However, direct generation
from these distributions is often difficult, and~\cite{plummer15} suggested 
using
``nested MCMC''  as in Algorithm~\ref{alg:nestedmcmc} below.
Other methods for cut posterior evaluation are discussed by \cite{liu+g20}, \cite{jacob2020unbiased} and  
\cite{pompe+j21}.
 
\begin{algorithm}
	\caption{Nested MCMC Sampler for Cut Posterior}
	\label{alg:nestedmcmc}
	\begin{algorithmic}
		\State Generate sample $\etavec_1^{(1)},\ldots,\etavec_1^{(S)} \sim p_{\text{cut}}(\etavec_1|{\cal D})\propto p(\etavec_1)g_1({\cal D}|\etavec_1)$ using an MCMC scheme
		\For{$s=1,\ldots,S$}
		\State Generate single value $\etavec_2^{(s)}\sim p(\etavec_2|\etavec_1^{(s)},{\cal D})$ using an MCMC scheme
		\EndFor
	\end{algorithmic}
\end{algorithm}

\subsection{Variational cut posterior computation}
Given the difficulty of exact cut posterior computation, 
variational inference methods to do so have been suggested
by \cite{yu+ns21} and~\cite{carmona+n22}.  Lemma 1 of \cite{yu+ns21}
establishes that the cut posterior distribution is closest
in Kullback-Leibler divergence to the true posterior
amongst distributions that have $\etavec_1$ marginal density
$p_{\text{cut}}(\etavec_1|{\cal D})$.  Therefore, 
if the family of approximations ${\cal F}$ is restricted to those that have $\etavec_1$ marginal
density $p_{\text{cut}}(\etavec_1|{\cal D})$, solving the conventional
variational optimization problem at~\eqref{eq:viprob} will also provide
the optimal variational approximation to the cut posterior. Crucially, solving 
this optimization does not require computation of the
intractable term $\bar{g}_2({\cal D}|\etavec_1)$ which creates
the computational bottleneck in MCMC.


This observation motivates a sequential VI procedure suggested by~\cite{yu+ns21}.
In a first stage an approximation of $p_{\text{cut}}(\etavec_1|{\cal D})$
is computed, which is then kept fixed in a second stage.
Consider a family of densities
of the form
$q_{\lambda}(\etavec)=q_{\widetilde{\lambda}}(\etavec_1)q_{\breve{\lambda}}(\etavec_2|\etavec_1)$,
where $\bm{\lambda}=(\widetilde{\bm{\lambda}}^\top,\breve{\bm{\lambda}}^\top)^\top$
are variational parameters partitioned into two sets.  The first set $\widetilde{\bm{\lambda}}$
parametrize the $\etavec_1$ marginal density, and the second set $\breve{\bm{\lambda}}$ parametrize the conditional density for $\etavec_2|\etavec_1$.  
If $D_{KL}\left(q\,||\,p\right)$ denotes the KL divergence of $q$ from $p$, then Algorithm~\ref{alg:vigeneral} below outputs an approximation 
to the joint cut posterior.

\begin{algorithm}
	\caption{VI for Cut Posterior}
	\label{alg:vigeneral}
	\begin{algorithmic}
	\State 1. Select a fixed form variational approximation $q_{\lambda}(\etavec_1,\etavec_2)=q_{\widetilde{\lambda}}(\etavec_1)q_{\breve{\lambda}}(\etavec_2|\etavec_1)$
	\State   2. Solve the optimization 
	$$\widetilde{\bm{\lambda}}^*=\arg \min_{\widetilde{\lambda}} D_{KL}\left( q_{\widetilde{\lambda}}(\etavec_1)||p_{\text{cut}}(\etavec_1|{\cal D})\right)$$
	to obtain an approximation $q_{\widetilde{\bm{\lambda}}^*}(\etavec_1)$ of $p_{\text{cut}}(\etavec_1|{\cal D})$
	\State 3. Solve the optimization
	$$\breve{\bm{\lambda}}^*=\arg \min_{\breve{\lambda}} 
	D_{KL}\left( q_{\widetilde{\lambda}^*}(\etavec_1)q_{\breve{\lambda}}(\etavec_2|\etavec_1)||
	p(\etavec_1,\etavec_2|{\cal D})\right)$$
	to obtain an approximation $q^*(\etavec)=q_{\widetilde{\lambda}^*}(\etavec_1)q_{\breve{\lambda}^*}(\etavec_2|\etavec_1)$ of $p_{\text{cut}}(\etavec_1,\etavec_2|{\cal D})$
	\end{algorithmic}
\end{algorithm}


In our empirical work, Gaussian variational approximations are used with Gaussian density $q_\lambda(\bm{\eta})=\phi_N(\bm{\eta};\bm{\mu},\Sigma)$, with mean $\bm{\mu}$ and variance $\Sigma=LL^\top$, where is $L$ a lower triangular Cholesky factor.  
Partitioning
$\muvec$ and $L$ to be conformable with  $\bm{\eta}=(\etavec_1^\top,\etavec_2^\top)^\top$, so that
$$\bm{\mu}=\left[\begin{array}{c}
\bm{\mu}_{\eta_1} \\
\bm{\mu}_{\eta_2} 
\end{array}\right], \;\;\;\;
L=\left[\begin{array}{cc}
  L_{\eta_1} & \bm{0} \\
  L_{\eta_1,\eta_2} & L_{\eta_2} 
  \end{array} \right],$$
  then $q_{\widetilde{\lambda}}(\etavec_1)$ and $q_{\breve{\lambda}}(\etavec_2|\etavec_1)$
  are both Gaussian densities with parameters
   $\widetilde{\lambdavec}=(\bm{\mu}_{\eta_1}^\top,
\text{vech}(L_{\eta_1})^\top)^\top$, and
$\breve{\lambdavec}=(\bm{\mu}_{\eta_2}^\top,\text{vec}(L_{\eta_1,\eta_2})^\top,\text{vech}(L_{\eta_2})^\top)^\top$, where `vec' and `vech' are the vectorization and 
half-vectorization matrix operators, respectively.
Methods for optimizing a Gaussian variational density parametrized
by a Cholesky factor are well-known in the literature 
(e.g. \citet{titsias2014doubly}, among many others) and
we do not describe this in detail here.  Other fixed form approximating families
can also be used in this framework.

\section{Cutting Feedback for Misspecified Copulas}\label{sec:cutting1}
A copula model can be viewed as a two module system, where the marginals $F_1(\cdot;\thetavec_1),\ldots,F_m(\cdot;\thetavec_m)$ form one module, and the copula $C(\cdot;\psivec)$ is a second module. 
This section discusses cutting feedback when the marginals are thought to be 
adequate, but the
copula function may be misspecified. We label the cut posterior for this case ``type 1'' to distinguish it from that in Section~\ref{sec:cutting2}.

\subsection{Type 1 cut posterior specification}
For the copula model with density at~\eqref{eq:copden}, we set
$\etavec_1=\thetavec$ and $\etavec_2=\psivec$ and factor the 
likelihood as
$g({\cal D}|\bm{\theta},\bm{\psi})=g_1({\cal D}|\bm{\theta})g_2({\cal D}|\bm{\theta},\bm{\psi})$,
where
$$g_1({\cal D}|\bm{\theta})=\prod_{i=1}^n \prod_{j=1}^m f_j(y_{ij};\bm{\theta}_j),\;\;
\mbox{ and }\;\;
g_2({\cal D}|\bm{\theta},\bm{\psi})=\prod_{i=1}^n 
c(F_1(y_{i1};\bm{\theta}_1),\dots, F_m(y_{im};\bm{\theta}_m);\bm{\psi}).$$
Assuming prior density $p(\bm{\theta},\bm{\psi})=p(\bm{\theta})p(\bm{\psi})$, 
with $p(\bm{\theta})=\prod_{j=1}^m p(\bm{\theta}_j)$, then the marginal cut posterior
at  \eqref{cut-theta-marginal}  simplifies to
$$p_{\text{cut}}(\bm{\theta}|{\cal D})=\prod_{j=1}^m p_j(\bm{\theta}_j|\bm{y}_{(j)}),$$
where $\bm{y}_{(j)}=(y_{1j},\dots, y_{nj})^\top$ denotes 
the data for the $j$th marginal and   
$p_j(\bm{\theta}_j|\bm{y}_{(j)})\propto p(\bm{\theta}_j)\prod_{i=1}^n f_j(y_{ij};\bm{\theta}_j).$

The ordinary and cut conditional posterior density is
$$p(\bm{\psi}|\bm{\theta},{\cal D})=
\frac{p(\bm{\psi})\prod_{i=1}^n c(F_1(y_{i1};\bm{\theta}_1),\dots, F_m(y_{im};\bm{\theta}_m);\bm{\psi})}{\bar{g}_2({\cal D}|\bm{\theta})},$$
where $\bar{g}_2({\cal D}|\bm{\theta})= \int p(\bm{\psi})\prod_{i=1}^n c(F_1(y_{i1};\bm{\theta}_1),\dots, F_m(y_{im};\bm{\theta}_m);\bm{\psi}) \,d\bm{\psi}$.
The joint cut posterior is 
\begin{equation}
p_{\text{cut}}(\bm{\theta},\psivec|{\cal D})=p_{\text{cut}}(\bm{\theta}|{\cal D})p(\bm{\psi}|\bm{\theta},{\cal D})\,,\label{eq:jcutpost1}
\end{equation} 
and we consider its computation using both the nested MCMC and variational approaches  in Algorithms~\ref{alg:nestedmcmc} and~\ref{alg:vigeneral}. 

\subsection{Theoretical equivalence with IFM}\label{sec:theory_cut1}
Among the most popular methods for estimating copula models is the  ``inference for margins''(IFM) procedure of~\cite{joexu1996} and~\cite{joe2005}. In IFM, each $\bm{\theta}_j$ is estimated by maximizing the likelihood of the $j$th marginal model, and then the copula parameters $\bm{\psi}$ are estimated by maximizing the likelihood conditional on these estimates. We now 
show for large $n$ the cut posterior
at~\eqref{eq:jcutpost1}
resembles a Bayesian version of IFM.

We first establish that the posterior mean of $p_{\text{cut}}(\bm{\theta},\bm{\psi}|{\cal D})$, denoted as $\bar\be:=\int \etavec p_{\text{cut}}(\be|{\cal D}) d\be$, is asymptotically equivalent to the IFM point estimator. To this end, define $\widehat{\bt}$ as the IFM estimator obtained by first maximizing $\log g_1(\MD|\bm{\theta})$, define $\widehat{\bm{\psi}}$ as the IFM estimator obtained by maximizing $\log g_2(\MD\mid \bp,\widehat{\bt})$ over $\bp$, and set $\widehat{\be}=(\widehat{\bt}^\top,\widehat{\bp}^\top)^\top$. Lemma~\ref{lem:ifm1} below establishes that the IFM point estimator and the cut posterior mean are asymptotically equivalent.
\begin{lemma}\label{lem:ifm1}
	If Assumptions~\ref{ass:cons} and~\ref{ass:dist2} in Part~\ref{app:dtf} of the Web Appendix are satisfied, then
	$
	\sqrt{n}(\overline{\be}-\widehat\be)=o_p(1).
	$
\end{lemma}
\noindent Assumptions \ref{ass:cons} and~\ref{ass:dist2} are similar to the standard regularity conditions employed in two-step copula modeling to deduce asymptotic normality of the IFM point estimator in \cite{joe2005}; see Part~\ref{app:dtf} of the Web Appendix for their specification and a 
detailed discussion. 

Lemma~\ref{lem:ifm1} does not address the accuracy with which the cut posterior quantifies uncertainty. To establish this we require the following additional definitions and observations.   
Let $P_0$ denote the true data generating process (DGP) for the observed data, and $p_0$ its density, then under Assumptions \ref{ass:cons} and \ref{ass:dist2} in Part~\ref{app:dtf} of the Web Appendix, it can be shown that both $\bar{\bt}$ and $\widehat{\bt}$ are consistent estimators of 
$$\bt_0=\argmin_{\bt}D_{\text{KL}}\left(p_0\,||\,g_1(\cdot\mid\bt)\right)\,.
$$
That is, $g_1(\cdot\mid\bt_0)$ is the closest element of the class $\{\bt: g_1(\cdot\mid\bt)\}$ to  $P_0$ in terms of KL divergence, and $\bt_0$ is the corresponding pseudo-true value. 
Further, define the following matrix of second derivatives for the marginal model parameters (i.e., $\bt$) :
\begin{flalign*}
	\mathcal{I}=-\lim_{n\rightarrow+\infty}n^{-1}\E\left(\nabla^2_{\bt\bt}\log g_1(\MD\mid\bt_{0})\right)\,, 
\end{flalign*}
where $\E$ is the expectation with respect to $P_0$.
Then Lemma~\ref{lem:two} below shows how the cut posterior for $\bt$ quantifies uncertainty.
\begin{lemma}\label{lem:two}
	If Assumptions \ref{ass:cons} and \ref{ass:dist2} in Part~\ref{app:dtf} of the Web Appendix are satisfied, then
	$$
\int\left|p_{\cut}(\bt|\MD)
-\phi_N\left(\bt;\widehat\bt,[n\mathcal{I}]^{-1}\right)\right| d\bt
=o_p(1). $$ 	
\end{lemma}
This result shows that asymptotically the cut posterior for $\thetavec$ resembles a Gaussian distribution centred at the IFM $\widehat\bt$, and with variance $\mathcal{I}^{-1}/n$. Therefore, when the marginals are correctly specified, the type 1 cut posterior for $\bt$ correctly quantifies uncertainty\footnote{By this we mean that a level $(1-\alpha)$ credible set asymptotically has frequentist coverage at the $(1-\alpha)$ level under $P_0$; i.e., Bayesian credible sets agree asymptotically with frequentist confidence sets.} for the unknown parameter value $\bt_0$, even if the copula function is misspecified. That is, $p_{\cut}(\bt|\MD)$ delivers inferences that are asymptotically the same as IFM and also correctly quantifies uncertainty.  

To understand how the marginal cut posterior for $\bp$, 
$$p_\cut(\bp|\MD)=\int p_{\text{cut}}(\bm{\theta},\psivec|{\cal D})d\bt=\int p_{\text{cut}}(\bm{\theta}|{\cal D})p(\bm{\psi}|\bm{\theta},{\cal D})d\bt,$$
quantifies uncertainty, define 
$$\bp_0=\argmin_{\bp}D_{\text{KL}}\left({p^{}_0}\,||\,g_2(\cdot\mid\bp,\bt_0)\right),
$$which is the pseudo-true value for the copula parameters $\bp$ when the unknown $\bt$ is replaced by $\bt_0$, and define the following matrix of second derivatives:
\begin{flalign*}
\mathcal{M}=\begin{pmatrix}	\mathcal{M}_{\bt\bt}&\mathcal{M}_{\bt\bp}\\\mathcal{M}_{\bp\bt}&\mathcal{M}_{\bp\bp}
\end{pmatrix},\text{ where }\mathcal{M}_{\bt\bp}=-\lim_{n\rightarrow+\infty}n^{-1}\E\left(\nabla^2_{\bt\bp}\log g_2(\MD\mid\bt_{0},\bp_0)\right),
\end{flalign*}
with $\mathcal{M}_{\bt\bt}$ and $\mathcal{M}_{\bp\bp}$ defined analogously; then the following lemma holds.
\begin{lemma}\label{lem:ifm2}
	If Assumptions \ref{ass:cons} and \ref{ass:dist2} in Part~\ref{app:dtf} of the Web Appendix are satisfied, then
	$$
\int\left|p_{\cut}(\bp|\MD)-
\phi_N\left(\bp;\widehat{\bp},n^{-1}(\mathcal{M}_{\bp\bp}^{-1}+\mathcal{M}_{\bp\bp}^{-1}\mathcal{M}_{\bp\bt}\mathcal{I}^{-1}\mathcal{M}_{\bt\bp}\mathcal{M}_{\bp\bp}^{-1})\right)\right| d  \bp=o_p(1).
	$$	
\end{lemma}

Lemma \ref{lem:ifm2} demonstrates that the variability for the marginal cut posterior of $\bp$ depends on the variability of the cut posterior for $\bt$. Hence, uncertainty flows from $p_\cut(\bt|\MD)$ to $p_\cut(\bp|\MD)$, but not the other way. This implies that  if the cut posterior for  $\bp$ is to correctly quantify uncertainty, then both the marginal components and the copula function must be well-specified.

\subsection{Simulation~1}
\label{sec::simExample1}
A simulation study compares the accuracy of the type 1 cut posterior to that of the conventional (i.e. uncut) posterior and IFM. For each sample size $n \in \{100, 500, 1000\}$ a total of $S=500$ datasets are generated from a bivariate copula model. The marginal $f_1$ is a log-normal distribution, with mean and variance parameters $\mu=1$ and $\sigma^2=1$, while the marginal $f_2$ is a gamma 
distribution with shape and rate parameters $\alpha=7$ and $\beta=3$. A t-copula \citep{demarta2005} is used with Kendall's tau $\tau=0.7$ and unity degrees of freedom parameter.

For each dataset, we fit a
copula model with the correct marginal
distributional forms, along with a bivariate Gumbel
copula.
Thus, the copula is misspecified, but can still capture correlation, as measured by Kendall's tau $\tau$, equal to that of the DGP. 
We assign vague proper
priors $\mu \sim N(0,100^2)$, $\sigma^2 \sim \text{Half-Normal}(0,100^2)$, $\alpha \sim \text{Half-Cauchy}(0,5)$, $\beta \sim \text{Half-Cauchy}(0,5)$, and $\tau  \sim \text{Uniform}(0,1)$, where $\text{Half-Cauchy}(m,s)$ 
is a half Cauchy distribution with location $m$ and scale $s$.
Both the outlined variational methodology and MCMC algorithms are used, with details given in Part~A1 of the Web Appendix. This results in four Bayesian posteriors, the means of which are used as point estimators. IFM is also used for comparison.

\begin{table}[htbp]
\renewcommand\arraystretch{1.15}
\caption{Parameter Point Estimation Accuracy in Simulation 1 ($n=1000$)} 
\label{tab:sim1biasrmse}    
\begin{center}
\begin{tabular}{ccccccccc}
            \hline\hline
            & \multicolumn{5}{c}{Misspecified Copula Fit} & &\multicolumn{2}{c}{Correct Copula Fit} \\ \cline{2-6}\cline{8-9} 
            & Uncut/ & Cut/ & IFM & Uncut/ & Cut/ & &Uncut/ &Cut/\\
            & MCMC & MCMC & &VI &VI & &MCMC &MCMC\\ \hline
            Parameter &{\em Bias} & & & & & & &\\ \cline{2-6}\cline{8-9} 
            $\mu$ & 0.0102 & {\bf 0.0019} &  0.0039 & 0.0089 & {\bf 0.0019} & &0.0014 &0.0019\\ 
            $\sigma^2$ & 0.0360 & 0.0069 &  {\bf 0.0012} & 0.0386 & 0.0288  & &0.0057 &0.0069\\ 
            $\alpha$ & -0.1803 & -0.0073 & 0.0137  & -0.1869 & {\bf -0.0033} & &-0.0035 &-0.0073\\
            $\beta$ & -0.0924 & -0.0046 & 0.0045  & -0.0938 & {\bf -0.0028} & &-0.0019 &-0.0046\\
            $\tau$ & 0.0192 & 0.0107 & 0.0139  & 0.0176 & {\bf 0.0090} & &0.0002 &-0.0101\\
            &{\em RMSE} & & & & & & &\\ \cline{2-6}\cline{8-9}
            $\mu$ & 0.0327 & {\bf 0.0295} & 0.0296  & 0.0298 & 0.0307 & & 0.0263 & 0.0295\\
            $\sigma^2$ & 0.0612 & 0.0472 &  {\bf 0.0333}  & 0.0620 & 0.0560 & & 0.0447 & 0.0472\\
            $\alpha$ & 0.3774 & {\bf 0.3114} & 0.3132  & 0.3580 & 0.3114 & & 0.3048 & 0.3114\\
            $\beta$ & 0.1689 & {\bf 0.1406} & 0.1413  & 0.1644 & 0.1407 & & 0.1395 & 0.1406\\
            $\tau$ & 0.0244 & 0.0181 & 0.0209  & 0.0232 & {\bf 0.0175} & & 0.0132 & 0.0182\\ \hline\hline
\end{tabular}
\end{center}
The bias and RMSE values of different point estimators computed over the $S=500$ simulation replicates, with the lowest values
in bold. Results on the left are where the misspecified copula is fit using the posterior mean from the conventional (i.e. uncut) and cut (type 1) posteriors, computed exactly using MCMC or approximately using VI. IFM is included for comparison. Results on the right are where the correct copula is fit using the 
posterior mean from the conventional (i.e. uncut) and cut (type 1) posteriors
computed exactly using MCMC. 
\end{table}

To measure estimation accuracy of the true parameter
values in the DGP, the bias and root mean square error (RMSE) is evaluated over the $S$ replicates. Table~\ref{tab:sim1biasrmse} (left-hand side) reports these for the case where $n = 1000$, and we make four observations. First, the cut posterior has lower bias and RMSE than the conventional posterior for all parameters, so that cutting feedback from the misspecified copula improves estimation accuracy. Second,  $\tau$ is estimated more accurately using its cut posterior than IFM. Third, the variational and exact posterior results are similar, suggesting the former is an accurate approximation. (Although, if a Gaussian VA underestimates uncertainty for other target posteriors, then a richer variational family can also be used.) Last, IFM provides a more accurate estimate of $\sigma^2$, but
is less accurate than the cut posterior for all other parameters.

We also consider accuracy when the correctly specified
copula model (i.e. a t-copula with the correct degrees of freedom) is fit. Table~\ref{tab:sim1biasrmse} (right-hand side) reports the bias
and RMSE for both the cut and conventional posterior in this case, both estimated
exactly using MCMC and the same uniform prior on $-1<\tau<1$. The 
cut posterior is only slightly less accurate than the conventional (i.e. uncut) posterior.

To assess the accuracy of the marginal posterior distributions we compute the coverage of their 95\% credible intervals. To assess
the accuracy of the point estimates of the copula model
components (in addition to their parameter values) we compute their 
predictive KL divergences. The latter is defined for marginal $j=1,2$ as
$$
\text{KL}_j = \int f_j (y ; \widehat{\thetavec}_j) \left [ \log  f_j (y ; \widehat{\thetavec}_j) - \log f_j^\star (y) \right ] \; \text{d} y,
$$
where $\widehat{\thetavec}_j$ is a point estimate of $\thetavec_j$, and $f_j^\star$ is the true marginal density of the DGP. 
The predictive KL divergence for the copula is defined as
$$
\text{KL}_{cop} = \int \int c (u, v ; \widehat{\tau}) \left [ \log  c (u, v ; \widehat{\tau}) - \log c^\star (u, v)  \right ] \; \text{d} u \text{d} v\,,
$$
where $\widehat{\tau}$ is a point estimate of $\tau$, and $c^\star$ is the true copula density for the DGP. The integrals above are computed numerically.

Table~\ref{tab:sim1covkl} reports the coverage probability of the credible intervals, along with the mean of the KL divergence metrics over the $S$ replicates. The results further confirm that under misspecification of the copula function (left-hand side of the table), the cut posterior is substantially more accurate than the conventional (uncut) posterior, and that the variational and exact posteriors are very similar.
Despite IFM estimating $\sigma^2$ slightly more
accurately than the cut posterior, the cut posterior either equals or out-performs estimation accuracy of all model components as measured by the KL divergences. Again, we see that when the correct model is fit (right-hand side of the table), the cut posterior 
is only slightly less accurate than the conventional posterior.

\begin{table}[htbp]
\caption{Posterior Coverage and Copula Model Accuracy in Simulation 1 ($n=1000$)} \label{tab:sim1covkl}

\begin{center}
\renewcommand\arraystretch{1.15}
\begin{tabular}{ccccccccc}
            \hline\hline
& \multicolumn{5}{c}{Misspecified Copula Fit} & &\multicolumn{2}{c}{Correct Copula Fit} \\ \cline{2-6}\cline{8-9} 
& Uncut/ & Cut/ & IFM & Uncut/ & Cut/ & &Uncut/ &Cut/\\
& MCMC & MCMC & &VI &VI & &MCMC &MCMC\\ \hline
Parameter &\multicolumn{3}{l}{\em Coverage Probabilities} & & & & &\\ \cline{2-6}\cline{8-9} 
            $\mu$ & 0.9600 & 0.9660 & - & 0.9760 & {\bf 0.9520} & &0.9510 &0.9660\\ 
            $\sigma^2$ & 0.8660 & 0.9380 & - & 0.9780 & {\bf 0.9560} & &0.9460 & 0.9380\\
            $\alpha$ & 0.8920 & {\bf 0.9560} & - & 0.8240 & 0.9760 & &0.9460 &0.9560\\
            $\beta$ & 0.8660 & {\bf 0.9380} & - & 0.9840 & 0.9680 & &0.9460 &0.9380\\
            $\tau$ & 0.5400 & 0.7440 & - & 0.7060 & {\bf 0.7780} & &0.9500 &0.9260\\ 
Component &\multicolumn{3}{l}{\em Mean Predictive  KL Divergence} & & & & &\\ \cline{2-6}\cline{8-9} 
            $f_1$ & 0.0014 & {\bf 0.0010} &  {\bf 0.0010} & 0.0014 & 0.0012 & & 0.0008 & 0.0010\\
            $f_2$ & 0.0014 & {\bf 0.0010} & 0.0013 & 0.0012 & {\bf 0.0010} & & 0.0010 & 0.0010\\ 
            $c$ & 0.1488 & 0.1424 & 0.1457 & 0.1499 & {\bf 0.1411} & & 0.0008 & 0.0012\\ \hline\hline
    \end{tabular}
\end{center}
Top: coverage probabilities for 95\% credible intervals for each parameter, with
those closest to 0.95 in bold. Bottom: the mean predictive KL divergence for the two marginals and the copula density for their estimate, with the lowest values 
in bold. 
Results on the left are where the misspecified copula is fit using the conventional (i.e. uncut) and cut (type 1) posteriors, computed exactly using MCMC or approximately using VI. IFM is included for comparison. Results on the right are where the correct copula is fit using the conventional (i.e. uncut) and cut (type 1) posteriors
computed exactly using MCMC. 
\end{table}     

Results for the cases where $n = 100$ and $n=500$ are reported
in Part~A4 of the Web Appendix, and are very similar to those for $n=1000$.

\section{Cutting Feedback for Misspecified Marginals}\label{sec:cutting2}
This section discusses cutting feedback when the copula function $C(\cdot;\psivec)$ is adequate, but the 
marginals $F_1(\cdot;\thetavec_1),\ldots,F_m(\cdot;\thetavec_m)$ are misspecified. 
We label the cut posterior for this case ``type 2'' and use a pseudo likelihood of
the rank data for its specification. Evaluation of this cut posterior
is more challenging than that in Section~\ref{sec:cutting1}, and to do so in higher dimensions we introduce an 
extension of this pseudo likelihood~\citep{PitChaKoh2006,hoff07,smith2012estimation} 
and then define a cut version of the resulting augmented posterior which is both tractable
and has the desired
type 2 cut posterior as its marginal.

\subsection{Type 2 cut posterior specification}
Setting $\etavec_1=\psivec$ and $\etavec_2=\thetavec$, to define the marginal
cut posterior for $\psivec$ we use a pseudo likelihood based on the rank data.
For each $y_{ij}$ define its rank within marginal $j$ as $r(y_{ij})$ \footnote{For example, in the absence of ties this is $r(y_{ij})=\sum_{k=1}^n \mathds{1}(y_{kj}\leq y_{ij})$.} and denote all the rank data as $r({\cal D})=\{r(y_{ij});i=1,\ldots,n, j=1,\ldots,m\}$. We employ 
the following probability mass function for the (discrete-valued) ranks 
\begin{equation}
	p_{\text{PL}}(r({\cal D})|\psivec):=\prod_{i=1}^n \Delta_{a_{i1}}^{b_{i1}}\cdots  \Delta_{a_{im}}^{b_{im}} C(\vvec;\psivec)\,,\label{eq:rlike}
\end{equation}
where $a_{ij}=(r(y_{ij})-1)/(n+1)$, $b_{ij}=r(y_{ij})/(n+1)$, $\vvec=(v_1,\ldots,v_m)^\top$\,, and where
\begin{equation*}
	\Delta_{a_{ij}}^{b_{ij}} C(\vvec;\psivec):=C(v_1,\ldots,v_{j-1},b_{ij},v_{j+1},\ldots,v_m;\psivec)-
	C(v_1,\ldots,v_{j-1},a_{ij},v_{j+1},\ldots,v_m;\psivec)\,,
\end{equation*}
is a differencing operator over element $j$~\citep[p.43]{nelsen06}. 
This is the likelihood under the assumption that each marginal is an empirical distribution function. It is related to the ``rank likelihood'' that is obtained from the exact distribution of the ranks; for example, see~\cite{hoff07} for specification of the rank likelihood of a Gaussian copula. However, as we discuss
later, it is more tractable than a rank likelihood.  It is also related to the popular pseudo-likelihood in~\cite{genest95}, but corrects for the discrete nature of the rank data.

The pseudo rank likelihood at~\eqref{eq:rlike} does not depend on the marginal parameters $\bt$ because the ranks are a strictly increasing transformation of $\MD$ and are unaffected by the marginal distributions. Therefore it can be used to define
a marginal cut posterior for $\psivec$ with density
\begin{equation}
p_{\text{cut}}(\psivec|{\cal D})=\frac{p_{\text{PL}}(r(D)|\psivec)p(\psivec)}{\int p_{\text{PL}}(r(D)|\psivec')p(\psivec')d\psivec'}\,.\label{eq:mcutpsi}	
\end{equation}

This definition fits into the two module system described in Section~\ref{sec:cfm} by considering
the factorization at~\eqref{eq:cflikefactor} with $g_1({\cal D}|\psivec)=p_{\text{PL}}(r(D)|\psivec)$ and $g_2({\cal D}|\psivec,\thetavec)=p({\cal D}|\psivec,\thetavec)/p_{\text{PL}}(r(D)|\psivec)$. With these definitions, the feedback
term is 
\[
\bar{g}_2({\cal D}|\psivec)=\int g_2({\cal D}|\psivec,\thetavec)p(\thetavec)d\thetavec=
\int \frac{p({\cal D}|\psivec,\thetavec)}{p_{\text{PL}}(r(D)|\psivec)}p(\thetavec)d\thetavec\,.
\]
In the two module system, the cut posterior at~\eqref{eq:mcutpsi} is obtained
by removing this feedback term. Notice that if the likelihood $p({\cal D}|\psivec,\thetavec)$ is close to the
pseudo rank likelihood $p_{\text{PL}}(r(D)|\psivec)$, then $\bar{g}_2({\cal D}|\psivec)\approx 1$ 
and the cut and ordinary posteriors
for $\psivec$ will also be close. Conversely, if the likelihood and the pseudo rank likelihood
deviate, the cut and ordinary posteriors will differ. 

The joint cut posterior is defined as
\begin{equation}
p_{\text{cut}}(\psivec,\thetavec|{\cal D})=p_{\text{cut}}(\psivec|{\cal D})p(\thetavec|\psivec,{\cal D})\,,
\label{eq:sec4jntcutpost}
\end{equation}
where the conditional $p(\thetavec|\psivec,{\cal D})=p({\cal D}|\thetavec,\psivec)p(\thetavec)/\int p({\cal D}|\thetavec',\psivec)
p(\thetavec') d\thetavec'$. The normalizing constant of this conditional is not computed when implementing
Algorithm~\ref{alg:vigeneral}.


\subsection{Theoretical behavior of the type 2 cut posterior}\label{sec:theory_cut2}
An advantage of the pseudo rank likelihood at~\eqref{eq:rlike} is that it is both computationally and 
theoretically more tractable than the rank likelihood of~\cite{hoff07} and others.  
As~\cite{hoff2014information} state,  the
rank likelihood ``is the integral of a
copula density over a complicated set defined by multivariate order constraints'', making it  intractable and complicating the derivation of theoretical results for parameter inference.
For example,~\cite{hoff2014information} control an accurate approximation of the rank likelihood in order to deduce their theoretical results.

In contrast, the type 2 cut posterior depends on~\eqref{eq:rlike} and the parametric likelihood $p(\MD|\bt,\bp)$, both of which are tractable. This allows direct analysis of the behavior of the cut posteriors $p_{\text{cut}}(\psivec|{\cal D})$ and $p_{\text{cut}}(\psivec,\bt|{\cal D})$ in \eqref{eq:mcutpsi} and~\eqref{eq:sec4jntcutpost}. To this end, let $M_n(\bp):=\log p_{\text{PL}}(r(\MD)|\bp)$, with $\mathcal{M}(\bp)=\lim_{n
\rightarrow \infty}M_n(\bp)/(1+n)$; further define $\widehat\bp_r=\argmax_{\bp} M_n(\bp)$, $\bp_\star=\argmax_{\bp}\mathcal{M}(\bp)$, and $\mathcal{M}_{\bp\bp}(\bp_\star)=-\nabla_{\bp\bp}^2\mathcal{M}(\bp_\star)$. Theorem~\ref{thm:ranks} below characterizes the behavior of $p_{\cut}(\bp|\MD)$.
\begin{theorem}\label{thm:ranks} 
If Assumptions \ref{ass:DGP1}-\ref{ass:crit1} in Part~\ref{app:rank} of the Web Appendix are satisfied, then
	$$
\int\left|p_{\cut}(\bp|\MD)-\phi_N\left(\bp;\widehat\bp_{r},n^{-1}\mathcal{M}_{\bp\bp}(\bp_\star)^{-1}\right)\right|d\bp=o_p(1).
$$ 
\end{theorem}	
Assumptions \ref{ass:DGP1}-\ref{ass:crit1} are given in  Part~\ref{app:rank} of the Web Appendix, 
and they 
ensure that $M_n(\bp)$ admits enough regularity so that the cut marginal posterior $p_{\cut}(\bp|\MD)$ satisfies a Berstein-von Mises result. 
While these assumptions are specific to the copula function, they are satisfied for popular choices,
including elliptical copulas, such as the student-t and Gaussian copulas, and key Archimedean copulas such as the Gumbel and Clayton copulas; see Part~\ref{sec:discuss2} of the Web Appendix for further discussion. 

Theorem~\ref{thm:ranks} implies that in large samples the cut posterior $p_{\cut}(\bp|\MD)$ based on the pseudo rank likelihood resembles a Gaussian density centered at $\widehat{\bp}_r$. If the copula is correctly specified, then the information matrix equality is satisfied and we have that $\mathcal{M}_{\bp\bp}(\bp_\star)^{-1}\equiv \mathrm{var}\left\{\nabla_{\bp} M_n(\bp_\star)/\sqrt{n}\right\}$, which has a particular form  given in Corollary~\ref{cor:two} in Web Appendix~\ref{app:dtf}. 
In such cases, Theorem \ref{thm:ranks} implies that the cut posterior based on the pseudo rank likelihood correctly quantifies uncertainty.

To state the behavior of $p_{\cut}(\bt|\MD)=\int p_{\text{cut}}(\psivec,\thetavec|{\cal D})d\bp$, let ${Q}_n(\bt,\bp)= \log p(\MD|\bt,\bp)$, and write $\mathcal{Q}(\bt,\bp):=\lim_{n\rightarrow \infty}n^{-1}\E\left(\log p(\MD|\bt,\bp)
\right)$, 
with derivatives of $\mathcal{Q}(\bt,\bp)$ denoted as $\mathcal{Q}_{ij}(\bt,\bp)=\nabla^2_{ij}\mathcal{Q}(\bt,\bp)$ for $i,j\in\{\bt,\bp\}$. Further define $\widehat\bt_r:=\argmax_{\bt} Q_n(\bt,\widehat\bp_r)$, $\bt_\star:=\argmax_{\bt}\mathcal{Q}(\bt,\bp_\star)$, 
$$
\Omega^{-1}=\mathcal{Q}_{\bt\bt}(\be_\star)^{-1}+\mathcal{Q}_{\bt\bt}(\be_\star)^{-1}\mathcal{Q}_{\bt\bp}(\be_\star)\mathcal{M}_{\bp\bp}(\bp_\star)^{-1}\mathcal{Q}_{\bp\bt}(\be_\star)\mathcal{Q}_{\bt\bt}(\be_\star)^{-1}\,,
$$ 
and $\be_\star=(\bt_\star^\top,\bp_\star^\top)^\top$. Then Theorem~\ref{thm:cut2} below characterizes the
behavior of $p_\cut(\bt|\MD)$.

\begin{theorem}\label{thm:cut2} 
If Assumptions \ref{ass:cons} and \ref{ass:dist2}  in Part~\ref{app:dtf} of the Web Appendix, and Assumptions \ref{ass:DGP1}--\ref{ass:crit1} in Part~\ref{app:rank} of the Web Appendix, and all satisfied, then
	$$
\int\left|p_{\cut}(\bt|\MD)-\phi_N\left(\bt;\widehat\bt_{r},n^{-1}\Omega^{-1}\right)\right|d\bt =o_p(1).
	$$
\end{theorem}	

Theorem~\ref{thm:cut2} shows that the uncertainty for the cut posterior of $\bt$ depends on the uncertainty in the cut posterior for $\bp$ through the term $\mathcal{M}_{\bp\bp}(\bp_\star)^{-1}$. Therefore, the cut posterior of $\bt$ will only  quantify uncertainty correctly if the copula model marginals and copula function are both well-specified. This is in contrast to the type 1 cut posterior where the marginal parameter posteriors delivered reliable uncertainty quantification as long as the marginal models were well-specified (see Lemma~\ref{lem:two}) .

\subsection{Simulation 2}
\label{sec::simExample3}
The simulation study in Section~\ref{sec::simExample1} is extended to compare the accuracy of the
type 2 cut posterior to that of the conventional posterior.
Data is generated from a bivariate copula model with the similar
marginals as in Simulation~1 (except that $\sigma^2 = 0.25$), but using a Gumbel copula with Kendall's tau $\tau=0.7$. For each dataset we fit a copula model with the correct copula family (i.e. a Gumbel), along with normal
marginals with mean and variance parameters $\mu_j,\sigma^2_j$
for $j=1,2$ and constrained to be positive. Thus, the marginals
are misspecified but the distribution has the same support as the DGP. We employ the vague proper priors $\mu_j, \sim N(0,100^2)$, $\sigma_j^2 \sim \text{Half-Normal}(0,100^2)$, and $\tau  \sim U(0,1)$. Both the outlined variational
methodology and MCMC algorithms are used to evaluate the type 2 cut posterior, along with the conventional posteriors, resulting in four Bayesian estimators. Details are given
in Part~A1 of the Web Appendix.

The accuracy of each posterior is measured using the predictive KL divergence metrics. 
Table~\ref{Simulation3KL} (left hand side) reports their mean values over the $S=500$ replicates 
for the 
case where $n=1000$. The cut posterior provides much more accurate estimates of both the marginal and copula components, compared to the conventional posteriors. Moreover, MCMC and variational estimates provide very similar levels of accuracy. Table~\ref{Simulation3KL} (right hand side) reports the
accuracy when the correctly specified copula model (i.e. with the correct
forms for the marginals)  is fit using both the conventional and cut posteriors computed using MCMC.
The same vague proper priors
are used for the marginal parameters, and the accuracy of the cut posterior is almost identical to the
conventional posterior. Results for the cases where $n=100$ and $n=500$ are reported in Part~A5 of the 
Web Appendix, and are very similar to those for $n=1000$.

\begin{table}[htbp]
\caption{Copula Model Estimation Accuracy in Simulation 2 ($n=1000$)} 	
\label{Simulation3KL}
\begin{center}
\renewcommand\arraystretch{1.15}
\begin{tabular}{lcccccccc}
	\hline \hline
	& \multicolumn{4}{c}{Misspecified Copula Fit} & &\multicolumn{2}{c}{Correct Copula Fit} \\ \cline{2-5}\cline{7-8} 
& Uncut/ & Cut/ & Uncut/ & Cut/ & &Uncut/ &Cut/\\
& MCMC & MCMC &VI &VI & &MCMC &MCMC\\ \hline
           Marginal $f_1$ & 0.8993 & {\bf 0.7945} & 0.8941 & 0.8090 &  &0.0003 & 0.0003 \\
          Marginal $f_2$ & 0.2103 & {\bf 0.1513} & 0.2034 & 0.1597 & & 0.0005 &0.0005 \\
           Copula & 0.0166 & {\bf 0.0010} & 0.0084 & {\bf 0.0010} & & 0.0007 & 0.0010\\
            \hline\hline
        \end{tabular}
\end{center}
Mean predictive KL divergence metrics for the two marginals and the copula density
of the bivariate copula model estimate. The lowest values are in bold. Results are given for the type 2 cut
and conventional (i.e. uncut) posteriors, computed exactly using MCMC or approximately using variational inference. Results are given for both the misspecified copula model (left hand side) and the correctly
specified copula model (right hand side).
\end{table}


\subsection{Augmented type 2 cut posterior}
When $m$ is small, the cut posterior can be evaluated by direct application of Algorithm~\ref{alg:vigeneral}. However, for even moderate values of $m$, 
the pseudo rank likelihood at~\eqref{eq:rlike} forms a computational bottleneck because it requires $O(n2^m)$ 
evaluations of $C$. In this case the computation can be avoided by employing
the extended likelihood  in \cite{smith2012estimation} which is tractable for higher values of $m$. 

Let  $\uvec_i=(u_{i1},\ldots,u_{im})^\top\sim C(\cdot;\psivec)$ and $\uvec=(\uvec_1^\top,\ldots,\uvec_n^\top)^\top$ be auxiliary variables, such that
$p(r({\cal D})|\uvec)=\prod_{ij}p(r(y_{ij})|u_{ij})=\prod_{ij}\mathds{1}(a_{ij}\leq u_{ij}<b_{ij})$. Then define an
extended likelihood as
\[
p(r({\cal D}),\uvec|\psivec):=p(r({\cal D})|\uvec)p(\uvec|\psivec)=
\prod_{ij}\mathds{1}(a_{ij}\leq u_{ij}<b_{ij})\prod_{i=1}^n c(\uvec_i|\psivec)\,.
\] 
Theorem~1 in~\cite{smith2012estimation} shows that integrating over
$\uvec$ retrieves the pseudo rank likelihood; i.e. $p_{\text{PL}}(r({\cal D})|\psivec)=\int p(r({\cal D}),\uvec|\psivec)d\uvec$.
Using this extended likelihood, we define the marginal cut posterior of $\psivec$ augmented with $\uvec$ as
\[
p_{\text{cut}}(\psivec,\uvec|{\cal D})=\frac{p(r(D),\uvec|\psivec)p(\psivec)}{\int p_{\text PL}(r(D)|\psivec')p(\psivec')d\psivec'}\,. 
\]
Integrating the density above over $\uvec$ gives the required cut posterior at~\eqref{eq:mcutpsi}.

Again, this setup fits into the  two module system discussed in Section~\ref{sec:cfm}, but with $\etavec_1=(\psivec^\top,\uvec^\top)^\top$ and $\etavec_2=\thetavec$, so that the cut posterior
\begin{equation}
p_{\text{cut}}(\psivec,\uvec,\thetavec|{\cal D})=p_{\text{cut}}(\psivec,\uvec|{\cal D})p(\thetavec|\psivec,{\cal D})\,,\label{eq:jaugcut}
\end{equation}
which we call the ``augmented cut posterior'' (i.e. the joint cut posterior augmented with $\uvec$). 
In this augmented cut posterior, $p(\thetavec|\psivec,\uvec,{\cal D})=p(\thetavec|\psivec,{\cal D})$ and
the marginal in $(\psivec^\top,\thetavec^\top)^\top$ is the required
cut posterior at~\eqref{eq:sec4jntcutpost}. We now discuss how to approximate~\eqref{eq:jaugcut} using recent developments in 
variational inference methods.  

\subsection{Variational inference for the augmented type 2 cut posterior}\label{sec:viaug}
The augmented cut posterior at~\eqref{eq:jaugcut} is estimated using Algorithm~\ref{alg:vigeneral} with approximation
\begin{equation}	q_\lambda(\thetavec,\psivec,\uvec)=q_{\widetilde{\lambda}}(\psivec,\uvec)
	q_{\breve{\lambda}}(\thetavec|\psivec)\,.\label{eq:cutva2}
\end{equation}
As before, a 
$N(\muvec,LL^\top)$ approximation is used
in $(\psivec^\top,\thetavec^\top)^\top$,
which has marginal in $\psivec$ with density  $q_{\widetilde{\lambda}_a}(\psivec)=\phi_N(\psivec;\muvec_\psi,L_{\psi}L_{\psi}^\top)$ and
parameters $\widetilde{\lambdavec}_a=(\muvec_\psi^\top,\text{vech}(L_\psi)^\top)^\top$. 
In Step~2 of the algorithm, $p_{\text{cut}}(\psivec,\uvec|{\cal D})$ is approximated by 
 $q_{\widetilde \lambda}$, for which we consider the family discussed below. 

For copula models with discrete-valued marginals, \cite{loaiza2019VBDA} study approximations to a posterior augmented by latents $\uvec$. They consider
VAs of the form $q_{\widetilde{\lambda}}(\psivec,\uvec)=
q_{\widetilde{\lambda}_a}(\psivec)q_{\widetilde{\lambda}_b}(\uvec)$ with  $\widetilde{\lambdavec}=(\widetilde{\lambdavec}_a^\top,\widetilde{\lambdavec}_b^\top)^\top$. They found approximations with marginal density in $\uvec$ given by
\[
q_{\widetilde{\lambda}_b}(\uvec)=\mathop{\prod_{i=1:n}}_{j=1:m} 
\frac{\phi_N(\zeta_{ij};\delta_{ij},\omega_{ij})}{(b_{ij}-a_{ij})\phi_N(\zeta_{ij};0,1)}\,,\;\; \zeta_{ij}=\Phi^{-1}\left(\frac{u_{ij}-a_{ij}}{b_{ij}-a_{ij}}\right)\,,
\]
provide a balance between scalability and accuracy. With this approximation, the
 parameters $\widetilde{\lambdavec}_b$ consist
of the $2nm$ mean and log-variance values $\{\delta_{ij},\log \omega_{ij};i=1,\ldots,n;\, j=1,\dots,m\}$.
 
This approximation is derived from adopting a normal
distribution for a transformation of $u_{ij}\in(a_{ij},b_{ij}]$ to the real line. An advantageous property is that $q_{\widetilde{\lambda}_b}$ can be shown to 
converge to the exact marginal cut posterior in $\uvec$ as $n\rightarrow \infty$, so that
for larger datasets it is a very accurate approximation. 
Another advantage of this approximation is that it is tractable, and fast to learn when combined with stochastic
gradient descent (SGD). We implement this optimization with control variates as outlined in \cite{loaiza2019VBDA},
where further details can be found.  

%

\section{Macroeconomic Example}\label{sec:macro}
Recent studies have applied high-dimensional copula 
models to multivariate economic and 
financial time series to capture both 
cross-sectional and serial dependence jointly; see~\cite{smith2015} 
and~\cite{nagler2022}
for examples. Copula models are attractive because
when the marginals are asymmetric, the 
predictive distributions exhibit time-varying asymmetry, which is an important
feature of such data. Out-of-sample density and tail forecasting are the primary objectives of these studies, for which heavy-tailed parametric marginals are preferred. 
To illustrate the impact of cutting feedback, we use it
to account for misspecification of 
either the marginals or copula function in such a model.

\subsection{Gaussian copula model}
We consider the Gaussian copula model of~\cite{smithvahey2016}, who
apply it to $N=4$ U.S.
macroeconomic time series observed quarterly, which are  $Y_{1,t}$  (Output Growth), $Y_{2,t}$ (Inflation), $Y_{3,t}$  (Unemployment Rate), 
and  $Y_{4,t}$ (Interest Rate). These four variables are observed at times
$t=1,\ldots,T$, so that 
the copula is of dimension $m=NT$, although $n=1$ because this is a single time series. 
The implicit copula of an $N$-dimensional stochastic
process $\{\bm{W}_t\}_{t=1}^T$ that follows
a stationary lag $p=4$ Gaussian vector autoregression (VAR) is used. It is a large Gaussian copula with
parameter matrix $\Omega$ that is a correlation matrix with
a sparse block Toeplitz structure. 

Rather than define 
the copula model likelihood directly in terms of $\Omega$, these authors express it more efficiently in terms of the unique semi-partial correlations. Appendix~\ref{app:a} shows how to do so, where
the unique semi-partial correlations 
associated with each lag $k=0,1,\ldots,p$ are grouped
together and denoted as
\begin{eqnarray}
\phivec(0) &= &\{\phi_{l_1,l_2}^0\} \mbox{ for } l_2=2,\ldots,N\;\text{and } l_2<l_1\,, \nonumber\\
\phivec(k) &= &\{\phi_{l_1,l_2}^k\} \mbox{ for } l_2=1,\ldots,N\;\text{and } l_2=1,\ldots,N\;\text{and }k=1,\ldots,p\,.
\label{eq:partials}
\end{eqnarray}
This is achieved by writing the Gaussian copula as a sparse D-vine where many of the component 
pair-copulas have density exactly equal to unity.
Denote $\phivec=\{\phivec(0),\ldots,\phivec(p)\}$ as the set of unique semi-partial correlations, then there is a one-to-one relationship between $\phivec$ and $\Omega$.

The original study considered quarterly data from 1954:Q1 until 2011:Q1. The data were 
sourced from the Federal Reserve Economic Database and the 2022:Q3 vintage from Real-Time Dataset for Macroeconomists hosted by the Philadelphia Federal Reserve.
In our analysis we extend the same economic time series to 2022:Q2, so that 
$T=274$. The matrix
$\Omega$ is of dimension $m=1096$, although 
it is parsimonious because the underlying copula process has only 72 unique semi-partial
 correlations $\phivec$. Regularization is known to improve the predictive performance of standard VAR models, so that
\cite{smithvahey2016} use a spike-and-slab prior on $\phivec$ for their copula model. 
In the current analysis, ridge priors with different levels of regularization at each lag are used. 
If $\widetilde{\phi}^k_{l_1,l_2}=\Phi^{-1}\left((\phi_{l_1,l_2}^k+1)/2\right)$ is a transformation
of $\phi_{l_1,l_2}^k$ to the real line, then the prior $\widetilde{\phi}^k_{l_1,l_2}\sim N(0,\tau^2_k)$ with $\tau^2_k \sim C^+(0,1)$ a half-Cauchy distribution.
The unconstrained copula and regularization parameters are therefore
$\psivec=\{\widetilde{\phivec},\log \tau_0^2,\log \tau_1^2,\ldots,\log \tau_p^2\}$.

In this application, prediction
of the distributional tails is necessary to quantify 
macroeconomic risk. 
A heavy-tailed parametric model is usually preferred to a non- or semi-parametric one because the latter tends to  under-weight the possibility of extreme events, such as that observed during the recent pandemic.
We follow the original study where time invariant
skew-t marginals were used for each variable (truncated to positive values for 
the Interest Rate variable). 
However, given the economic shocks since 2011:Q1, it is uncertain whether or not this choice of
marginals or Gaussian copula remain suitable for the extended dataset used here. Therefore, we consider cutting feedback first from the copula parameters
$\psivec$ to the marginal parameters $\thetavec$ (the type 1 cut posterior), and then also from $\thetavec$ to $\psivec$ (the type 2 cut posterior).
Because $m$ is large it is infeasible to compute the cut posteriors exactly,
and VI was used. For the type 2 cut posterior, the
variational approximation to the augmented posterior was employed as outlined in Section~\ref{sec:viaug}. When solving the variational optimizations, a SGD algorithm with ADADELTA learning rate was used with 2000 steps.

\subsection{Prediction and log-score metric}
To judge the accuracy of the different posteriors, we calculate a log-score metric using the posterior predictive
distribution as follows. If $\yvec_t=(y_{1,t},\ldots,y_{N,t})^\top$ is the observed value of the $N=4$ variables $\bm{Y}_t=(Y_{1,t},\ldots,Y_{N,t})^\top$
at time $t$, then the posterior predictive density $h$ steps ahead is
\begin{equation}
f_{t+h|t}(\yvec_{t+h}|\yvec_{t},\ldots,\yvec_{t-p+1})\equiv\int p(\yvec_{t+h}|\yvec_{t},\ldots,\yvec_{t-p+1},\thetavec,\psivec)\pi_t(\thetavec,\psivec)\mbox{d}\thetavec\mbox{d}\psivec\,, \mbox{ for } h\geq 1\,.\label{eq:postpred}
\end{equation}
Here, $\pi_t(\thetavec,\psivec)$ is a posterior density based on the data $y_1,\ldots,y_t$, for which we consider both variational cut posteriors and also the joint posterior. The integral is evaluated by averaging over 5000 draws from $\pi_t$. For the conventional posterior these are obtained
using an MCMC scheme as in~\cite{smithvahey2016} but where the regularization
parameters $\tau_0^2,\ldots,\tau_p^2$ are also drawn. Drawing from the variational cut posteriors is straightforward because they are fixed form Gaussian approximations.
Conditional on $(\thetavec,\psivec)$, draws from the predictive density $p(\yvec_{t+h}|\yvec_{t},\ldots,\yvec_{t-p+1},\thetavec,\psivec)$ can be
obtained using the sparse D-vine representation of the Gaussian copula as outlined in~\cite{smithvahey2016}.
\begin{figure}[thb]
	\centering
	\includegraphics[width=0.7\textwidth]{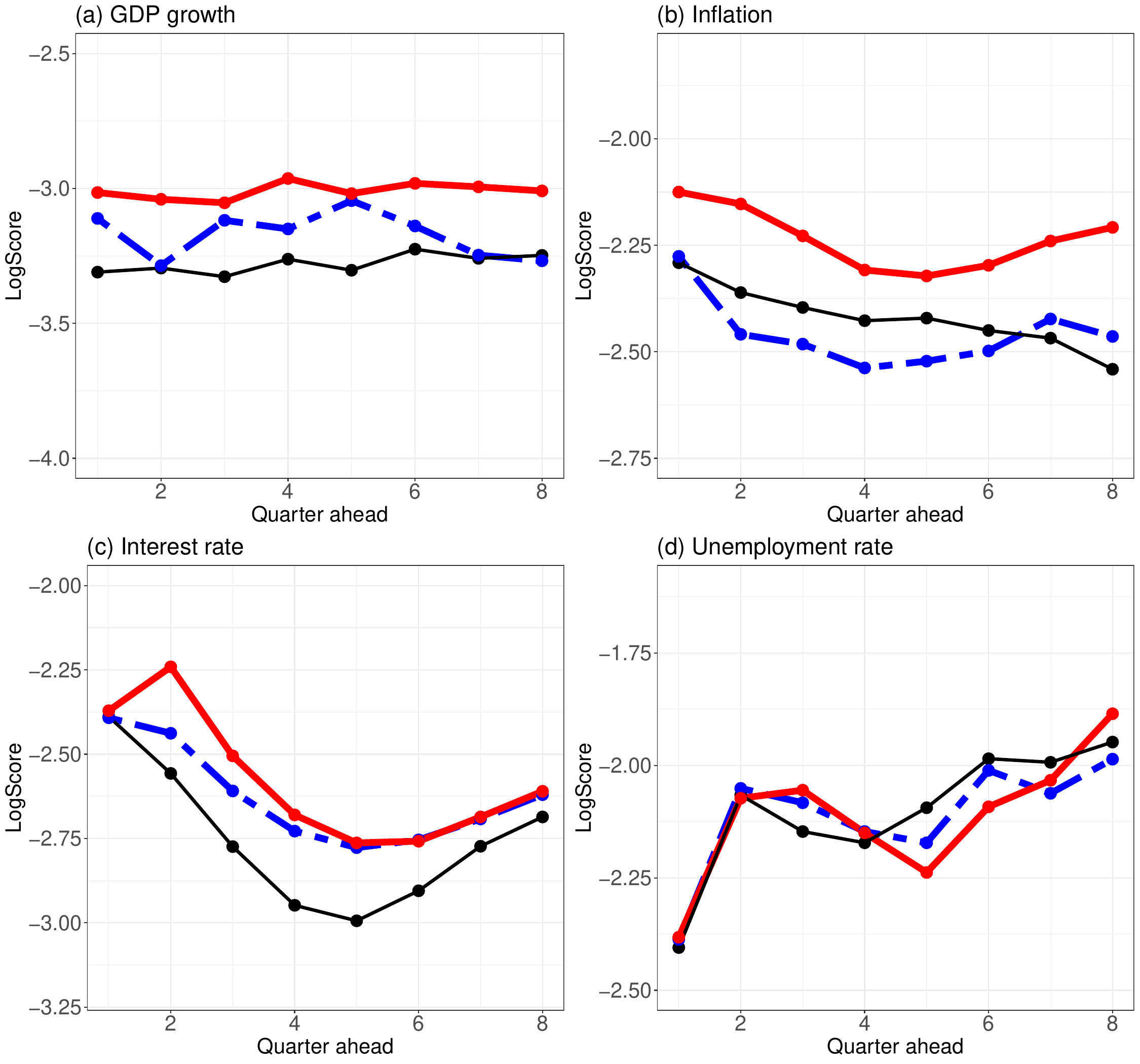}
	\caption{Plots of the log-score posterior predictive metric $LS_{j,h}$ for 
		the type~1 cut posterior (blue dashed line), type~2 cut posterior (red thick line) and 
		the conventional (i.e. uncut) posterior (black thin line). Panels (a-d) correspond to variables GDP Growth ($j=1$), Inflation ($j=2$), Interest 
		Rate ($j=3$) and Unemployment Rate ($j=4$), respectively. In each panel the metric values are plotted for predictions $h=1,\ldots,8$ quarters ahead. Higher values correspond to greater predictive accuracy.}
	\label{fig:macroLS}
\end{figure}

A log-score metric for variable $j$ predicted $h$ steps ahead can be computed
as 
\[
LS_{j,h}=\sum_{t=p}^{T-h} {\log \widehat{f_{t+h|t}}}(y_{j,t+h}|\yvec_{t},\ldots,\yvec_{t-p+1})\,.
\]
Here,
$\log\widehat{f_{t+h|t}}(y_{j,t+h}|\yvec_{t},\ldots,\yvec_{t-p+1})$ is a kernel density 
estimate of the logarithm of draws from~\eqref{eq:postpred}, evaluated at the observed value
$y_{j,t+h}$. 
Higher values of this log-score indicate 
better calibrated posterior distributions $\left\{\pi_p,\ldots,\pi_{T-h}\right\}$ for predictive purposes.

Figure~\ref{fig:macroLS} plots $LS_{j,h}$ for each variable $h=1,\ldots,8$ quarters
ahead, which matches the typical macroeconomic forecast horizon. By this metric, the 
type~2 cut posterior is a substantial improvement over the conventional posterior and type~1 cut posterior for GDP Growth (the main forecast variable), Inflation and the Interest Rate. This suggests that misspecification of the marginals impacts posterior inference, much more than any potential misspecification of the copula function. The approach of using predictive performance to select between 
cut and conventional posteriors has been discussed previously by~\cite{carmona+n22} in the context
of semi-modular inference.

\begin{figure}[thb]
	\centering
	\includegraphics[width=0.7\textwidth]{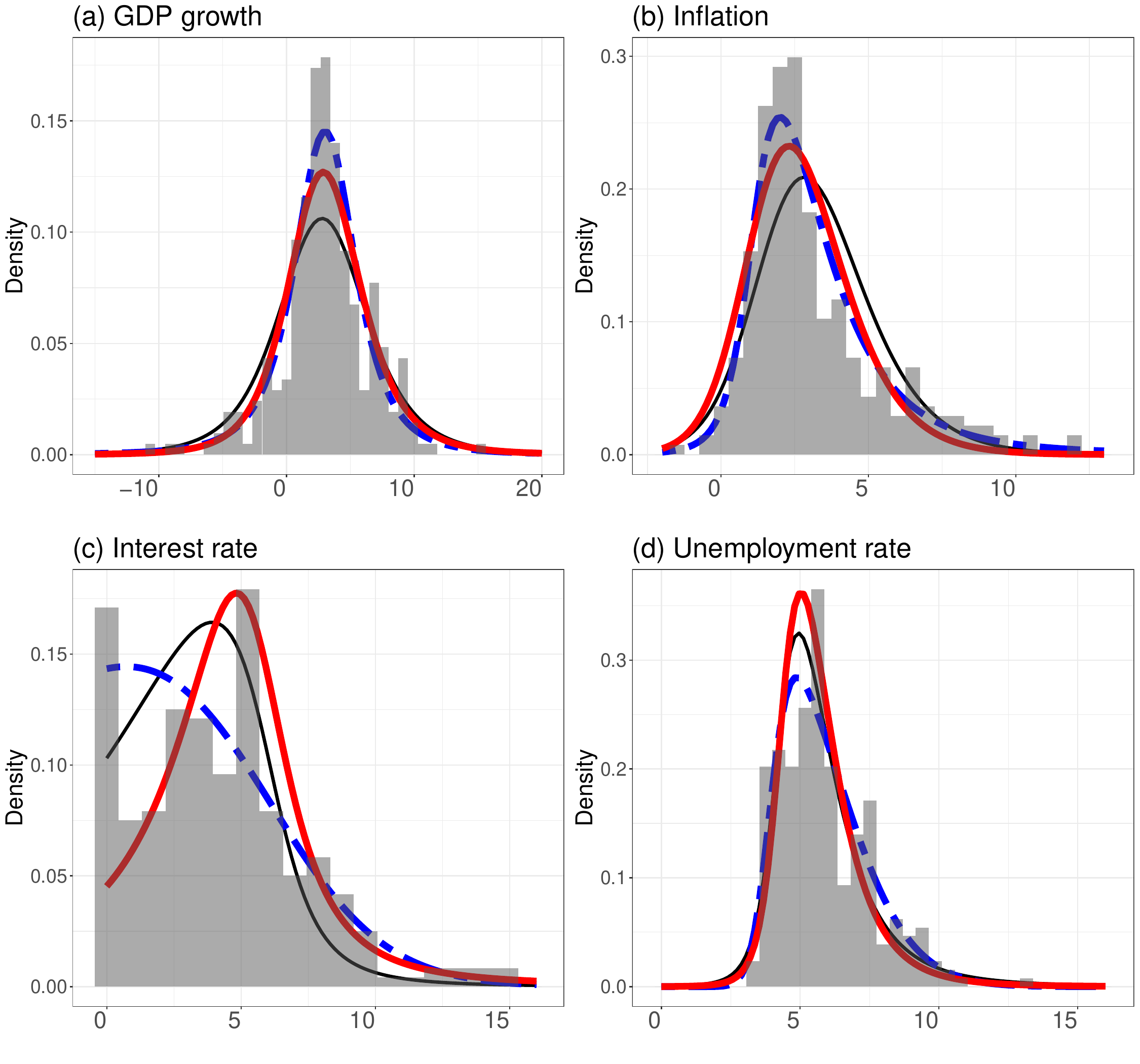}
	\caption{Plots of the estimated skew-t marginal densities for
		(a)~GDP Growth, (b)~Inflation, (c)~Interest Rate, and (d)~Unemployment Rate.
		These estimates are evaluated at the posterior means of $\thetavec$ for the type~1 cut posterior (blue dashed line), type~2 cut posterior (red thick line) and 
		the conventional (i.e. uncut) posterior (black thin line). Histograms of the data are also plotted.}
	\label{fig:macromargins}
\end{figure}

\subsection{Estimates}
Figure~\ref{fig:macromargins} 
plots
the estimated skew-t marginal densities for the four macroeconomic variables, along with histograms of the data. The three
posterior estimates differ substantially, highlighting the impact of cutting feedback
in this model. The  
histograms show that skew-t distributions are likely to be a misspecification 
for the copula model marginals in our extended dataset. 
For example, between 2011:Q1 and 2022, the Federal Reserve set interest rates to historical near-zero lows, corresponding to a mode at these values in
the histogram in panel~(c). 
While a truncated skew-t was an appropriate marginal for the pre-2011 data studied by~\cite{smithvahey2016}, it is inappropriate for the extended dataset that has a bimodal marginal in Interest Rate.
For this reason, the type~2 cut posterior correctly cuts feedback from the misspecified marginals when computing inference about 
the $\psivec$. This increases the overall accuracy of inference, as 
measured by the log-score metrics, relative to the conventional posterior.

\begin{figure}[p]
	\centering
	\includegraphics[height=0.9\textheight]{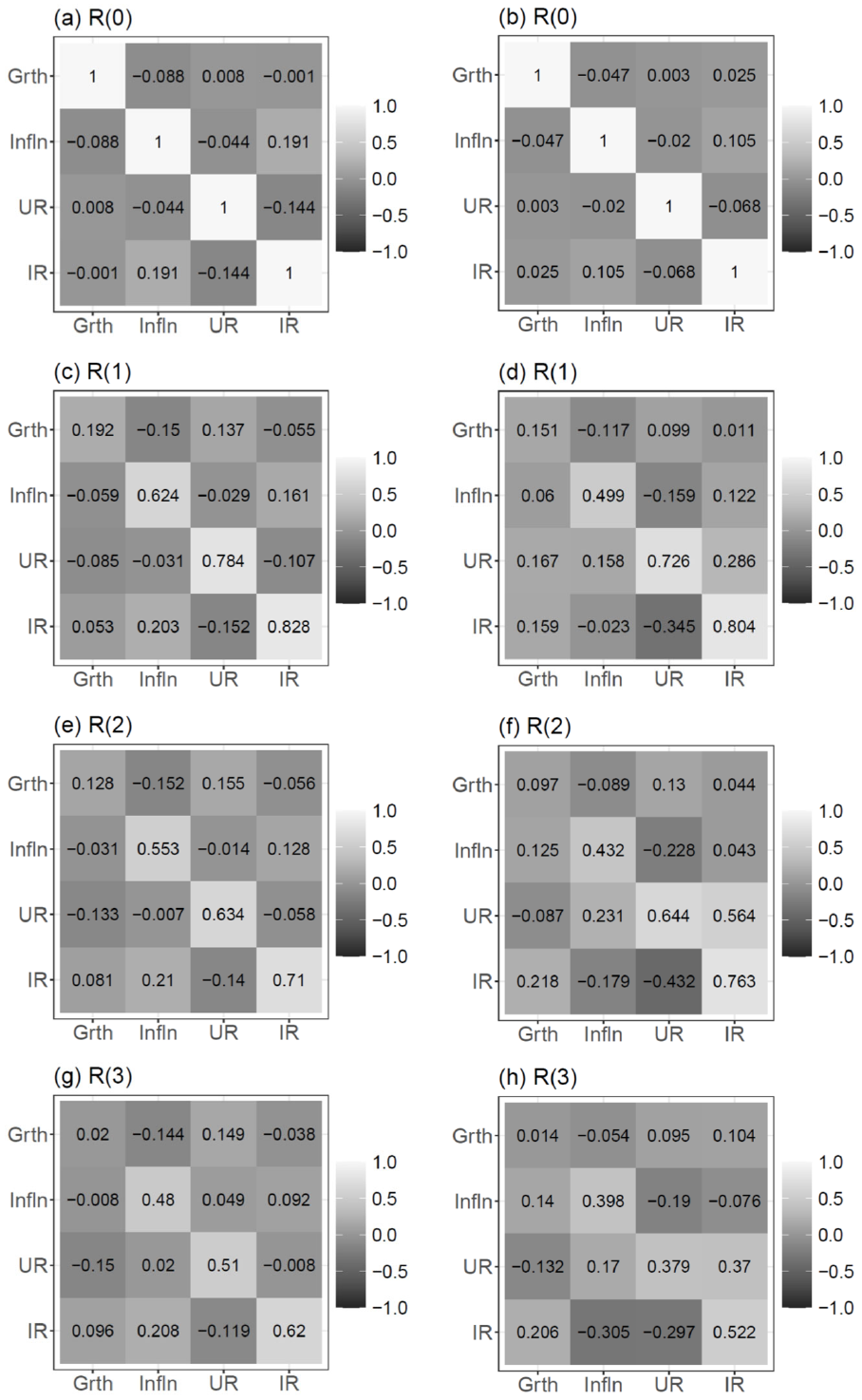}
	\caption{Posterior means of the matrices of Spearman pairwise correlations $R(k)$ for $k=0,1,2,3$ from the fitted copula model.  
	The left hand panels (a,c,e,g) contain results for the conventional (i.e. uncut) joint posterior, while the right hand panels 
	(b,d,f,h) contain results for the type~2 cut posterior. For example, the estimated Spearman correlation between
	the Interest Rate (IR) at time $t-3$ (variable $Y_{4,t-3}$) and  the Unemployment Rate (UR) at time $t$ (variable $Y_{3,t}$) is $-0.008$ 
	using the conventional posterior in panel~(g), and $0.370$ using the type~2 cut posterior in panel~(h).} 
	\label{fig:spearman}
\end{figure}

Finally, we consider the matrices $R(k)\equiv \{r_{i,j}(k)\}$ 
of pairwise Spearman's rho values $r_{i,j}(k)=\rho(Y_{i,t},Y_{j,t-k})$. These are a function
of the posterior of $\phivec$ as outlined in~\cite{smithvahey2016}, and their estimates provide important macroeconomic insights. 
Figure~\ref{fig:spearman} plots mean estimates of 
$R(0)$, $R(1)$, $R(2)$ and $R(3)$ using the conventional posterior (left hand panels), 
and using the type~2 cut posterior (right hand side). Cutting feedback 
perturbs these Spearman correlation estimates. For example, the pairwise correlation between 
the Interest Rate at time $t-3$ and Inflation at time $t$ is estimated to be $0.092$ in the conventional
 posterior, whereas in the type~2 cut posterior it is $-0.076$. The latter is more consistent with
monetary policy, where interest rate increases are often aimed at reducing future inflation. 
\section{Discussion}\label{sec:conc}
The modular nature of copula models can greatly simplify the specification of many multivariate stochastic models. It can also be used to improve the accuracy of statistical inference under potential model misspecification. As far as we are aware, this is the first paper to propose cutting feedback methods to do so.
We show theoretically and empirically that these methods can be more accurate in misspecified models than the conventional Bayesian posterior. 

Previous inference methods that control for  misspecification of the copula function when estimating the marginals include IFM~\citep{joexu1996,joe2005}.
For parametric marginals this is usually implemented using a two-stage maximum likelihood procedure, to which we show the type~1 cut posterior mean is asympototically equivalent. For nonparametric marginals, a 
well-established approach is to estimate the marginals using their empirical distribution functions, followed by estimating the copula parameters using pseudo-maximum likelihood; see~\cite{oakes1994} and~\cite{genest95}. This can be numerically unstable in higher dimensions, in which case kernel density estimators may be adopted for the marginals. If Bayesian nonparametric distributions~\citep{hjort2010bayesian} are used to model the marginals, then estimation using our proposed type~1 cut posterior provides a Bayesian equivalent which can be used in high dimensions when evaluated by variational methods. \cite{grazianliseo17} also
suggest Bayesian estimation of a copula model by generating each $\thetavec_j$ from their marginal posteriors, as at 
the first step of Algorithm~\ref{alg:nestedmcmc} when evaluating the type~1 cut posterior. However, they employ these draws to evaluate an approximate posterior of a dependence parameter based on an exponentially tilted likelihood, rather than a cut posterior.

\begin{table}[htbp]
\caption{Summary of Theoretical Results for Cut Posteriors}
\label{tab:theory}
\begin{center}	
\begin{tabular}{lcccccc}
		\toprule
		\multirow{2}{*}{} &
		\multicolumn{2}{c}{Both Correct} &
		\multicolumn{2}{c}{Marginal Miss. } &
		\multicolumn{2}{c}{Copula Miss. } 
		\\
		& {$\bt$} & {$\bp$} & {$\bt$}& {$\bp$} & {$\bt$}& {$\bp$}\\
		\midrule
  		Conventional Posterior& $\checkmark$ & $\checkmark$ & $\times$ & $\times$& $\times$ & $\times$\\
		Type 1 Cut Posterior& $\checkmark$ & $\checkmark$ & $\times$ & $\times$& $\checkmark$ & $\times$\\
			Type 2 Cut Posterior & $\checkmark$ & $\checkmark$ & $\times$ &$\checkmark$ & $\times$ & $\times$\\
		\bottomrule
	\end{tabular}
 \end{center}
In the column headings ``Both Correct" indicates that both  marginals and copula function are correctly specified; ``Marginal Miss." refers to the case where the marginals are misspecified, but the copula function is correct; and ``Copula Miss." refers to the case where the copula function is misspecified, but the marginals are correct. The parameters $\bt$ and $\bp$ refer to calibration of that specific parameter with ``$\checkmark$'' indicating (asymptotically) correct calibration, and ``$\times$'' denoting (possibly) inaccurate calibration.
\end{table}
Methods that control for misspecification of the marginals when estimating the copula function are rare, especially in high-dimensions. \cite{kim2007comparison} demonstrates that adopting nonparametric marginals as in~\cite{genest95} can guard against this, but this will be at the cost of reduced statistical efficiency when the marginals are in fact well-specified. Our type~2 cut posterior guards against this type of misspecification while attempting to limit any loss in statistical efficiency. Table~\ref{tab:theory} summarizes our theoretical results.
Along with the conventional posterior, both types of cut posterior are correctly calibrated asymptotically when the copula model is well-specified. However, unlike the conventional posterior, the type~2 cut posterior is also correctly calibrated under misspecification of the marginal models, and the type~1 cut posterior under misspecification of the copula function.

Evaluation of cut posteriors is difficult, and another contribution of our paper is the development of variational methods to do so for copula models. The definition of the type~2 cut posterior using 
a pseudo rank likelihood complicates computation, although this can be overcome by considering an augmented posterior with the cut posterior as its marginal in $(\thetavec,\psivec)$. Application of the variational methods to a 1096 dimension Gaussian copula for a macroeconomic forecasting application demonstrates their speed and efficiency in high dimensions. 

Finally, we note that macroeconomic example is also interesting in itself. Copula time series models have strong potential~\citep{smithmin2010,smith2015,smithman2018,nagler2022}, but selection of an appropriate copula function or parametric marginals can be difficult. In this case, guarding against misspecification is valuable, and our empirical work shows a cut posterior can increase density forecasting accuracy relative to the conventional posterior. Further useful applications of our new Bayesian methodology for cutting feedback in copula modeling await.

\newpage
\appendix
\setcounter{table}{0}
\setcounter{figure}{0}
\setcounter{algorithm}{0}
\renewcommand{\thetable}{\Alph{section}\arabic{table}}
\renewcommand{\thefigure}{\Alph{section}\arabic{figure}}
\renewcommand{\thealgorithm}{\Alph{section}\arabic{algorithm}}

\section{Copula multivariate time series model}\label{app:a}
This appendix gives the likelihood for the Gaussian copula model of~\cite{smithvahey2016}, to which 
we refer for full details.
Let the $T$ values of the VAR($p$) process be stacked into vector
$\bm{W}=(\bm{W}_1^\top,\ldots,\bm{W}_T^\top)^\top=(W_1,W_2,\ldots,W_m)^\top \sim N(0,\Omega)$. The VAR is stationary and constrained to have unit marginal variances,
so that 
\[
\Omega=
\left[ \begin{array}{ccc}
	\Omega(0) &\cdots &\Omega(T-1) \\
	\vdots &\ddots &\vdots \\
	\Omega(T-1) &\cdots &\Omega(0)
\end{array} \right]
\]
is a block Toeplitz correlation matrix with $\text{corr}(\bm{W}_t,\bm{W}_s)=\Omega(|t-s|)$.
For $i>j+1$, define the semi-partial correlation $\varphi_{i,j}=\mbox{corr}(W_i,W_j|W_{j+1},\ldots,W_{i-1})$ and
$\varphi_{i+1,i}=\text{corr}(W_{i+1},W_{i})$. 
There is a one-to-one transformation between $\Omega$ and the  
 semi-partial correlations $\varphivec=\{\varphi_{i,j}\}_{i=1:N,j<i}$ 
  due to Yule; e.g. see~\cite{daniels2009}. For a stationary VAR($p$) model, the
 majority of the elements in $\varphivec$ are either exactly zero or replicated values.
 \cite{smithvahey2016} show how to identify the unique values, which are denoted as $\phivec$ in Section~\ref{sec:macro}, and organize these
 into the blocks at~\eqref{eq:partials} that capture serial dependence at different lags.
 
Our copula model in Section~\ref{sec:macro} uses the implicit copula of $\bm{W}$, which is a Gaussian copula with parameter matrix
$\Omega$. It is well-known that a Gaussian copula can be written as a D-vine~\citep{czado2019} with 
density 
\begin{equation}
c(\uvec;\phivec)=\prod_{i=2}^m \prod_{j=1}^{i-1} c_{i,j}(u_{i|j+1},u_{j|i-1};\varphi_{i,j})\,,
\label{eq:dvine}
\end{equation}
where $c_{i,j}(\cdot,\cdot;\varphi_{i,j})$ is a bivariate Gaussian copula density with parameter
 $\varphi_{i,j}$ given by the semi-partial correlation defined above. When $\varphi_{i,j}=0$ the pair-copula is the 
independence copula with density $c_{i,j}(\cdot,\cdot;0)=1$. 
The arguments of each pair-copula, $u_{i|j+1}$ and $u_{j|i+1}$, 
can be computed from $\uvec$ and $\phivec$ efficiently using the recursive algorithm 
outlined in Appendix~A of~\cite{smithvahey2016}. This also gives an expression for  
the product at~\eqref{eq:dvine} in terms of only the non-independence pair-copula densities (i.e. those
pair-copula densities which are not equal to unity). Finally, because this model is for a single time series, the 
likelihood is simply given by~\eqref{eq:copden}.

\setcounter{section}{0}
\renewcommand{\thesection}{A\arabic{section}}
\section*{Web Appendix}
This Web Appendix has four parts:

\begin{itemize}
	\item[] {\bf Part~A1}: Additional computational details for the simulation studies.
	\item[] {\bf Part~A2}: Assumptions used to deliver the theoretical results in Section \ref{sec:theory_cut1}, and proofs of all stated results in Section \ref{sec:theory_cut1}.
	\item[] {\bf Part~A3}: Assumptions used to deliver the theoretical results in Section \ref{sec:theory_cut2}, and proofs of all stated results in Section \ref{sec:theory_cut2}.
	\item[] {\bf Part~A4}: Additional results for Simulation~1 when the sample size is $n=100$ and $n=500$.
	\item[] {\bf Part~A5}: Additional results for Simulation~2 when the sample size is $n=100$ and $n=500$.
\end{itemize}
\newpage

\section{Additional details for the simulation studies}

\subsection{Evaluation of conventional joint posterior}
In both Simulations~1 and~2, to compute the conventional (i.e. uncut) joint posterior exactly
using MCMC methods, we employed a Metropolis-Hastings (MH) sampler with a multivariate normal independence proposal centered at the posterior mode with covariance matrix given by the negative inverse Hessian evaluated at the mode.

\subsection{Evaluation of cut posteriors using MCMC}
Because the simulation examples are bivariate, it is possible to use MCMC to compute the cut posteriors using Algorithm~\ref{alg:nestedmcmc}. For the type~1 cut posterior in Simulation~1, in the first step we draw from the marginal cut posterior of $\thetavec$. To do so, a MH sampler is used with a normal approximation centred at the cut posterior mode as a proposal with 15,000 draws and
5,000 burn-in. In the second step, for each draw of
$\thetavec$, one draw of the copula parameter $\tau$ is obtained from its conditional  posterior using a second MH sampler with a normal approximation as a proposal and a burn-in of 100 iterations.

For the type~2 cut posterior in Simulation 2, in the first step we draw from the marginal cut posterior of $\tau$. To do so, a MH sampler is used with a
normal approximation centred at the cut posterior model as a proposal with 15,000 draws and 5,000 burn-in. In the second step, for each draw of $\tau$, one draw of the marginal parameters $\thetavec$ is obtained from their conditional posterior using 
a second MH sampler with a normal approximation as a proposal and a burn-in of 1000 iterations, 

\subsection{Evaluation of variational posteriors}
For the VI methods Gaussian  approximations to the cut and conventional posteriors are used as discussed in the manuscript. Algorithm~\ref{alg:vigeneral} is used to evaluate both the cut posterior in Section~\ref{sec:cutting1} and the 
cut augmented posterior in Section~\ref{sec:cutting2},
When solving all the variational optimizations, the optimal variational parameters are obtained using stochastic gradient descent (SGD). The 
ADADELTA automatic learning rate is used~\citep{zeiler12} with global learning rate $\rho = 0.85$ and conditioning perturbation $\epsilon = 10^{-6}$. The SGD algorithm is run for a conservative 10000 steps; see~\cite{ong+ns16} for further details on the application of SGD to variational optimization.
\newpage

\section{Assumptions and Results: Type~1 Cut Posterior}\label{app:dtf}
\subsection{Maintained Assumptions: Type~1 Cut Posterior}
The KL minimising point between the class $g_1(\MD\mid\bt)$ and the true density $p^{}_0$ can be defined as $g_1(\cdot\mid\bt_0)$, where
$$
\bt_0:=\argmin_{\bt\in\Theta}D_{\text{KL}}\left\{{p^{}_0}|| g_1(\cdot\mid\bt)\right\}.
$$Similarly, we can define the KL minimising point between the class $g_2(\MD\mid\bp,\bt_0)$ and $p^{}_0$ as $g_2(\cdot\mid\bp_0,\bt_0)$, where
$$
\bp_0:=\argmin_{\bp\in\Psi}D_{\text{KL}}\left\{{p^{}_0}|| g_2(\cdot\mid\bp,\bt_0)\right\}.
$$

Letting $\E(\cdot)$ denote expectation under $p^{}_0$,  define the following second derivative matrices:
\begin{flalign}
	\mathcal{I}:=-\lim_{n\rightarrow+\infty}\E\nabla_{\bt\bt}^{2}n^{-1}\log g_1(\MD\mid\bt_{0}),\quad \mathcal{M}:=-\lim_{n\rightarrow+\infty}\E\nabla_{\be\be}^{2}n^{-1}\log g_2(\MD\mid\be_{0}),
\end{flalign}where we can partition the matrix $\mathcal{M}$ as 
$$
\mathcal{M}=\begin{pmatrix}
	\mathcal{M}_{11}&\mathcal{M}_{12}\\\mathcal{M}_{21}&\mathcal{M}_{22}
\end{pmatrix},\text{ where }\mathcal{M}_{ij}:=-\lim_{n\rightarrow+\infty}\E\nabla_{\be_i\be_j}^{2}n^{-1}\log g_2(\MD\mid\be_{0}).
$$Lastly, define the matrix 
$$
V:=\begin{pmatrix}
	\mathcal{I}^{-1}&-\mathcal{I}^{-1}\mathcal{M}_{12}\mathcal{M}_{22}^{-1}\\-\mathcal{M}_{22}^{-1}\mathcal{M}_{21}\mathcal{I}^{-1}&\mathcal{M}_{22}^{-1}+\mathcal{M}_{22}^{-1}\mathcal{M}_{21}\mathcal{I}^{-1}\mathcal{M}_{12}\mathcal{M}_{22}^{-1}
\end{pmatrix}.
$$

The following smoothness assumptions are used to obtain the large sample results in Sections \ref{sec:theory_cut1}, and in particular to show that the IFM estimator and the type 1 cut posterior mean are asymptotically equivalent. We discuss these assumption in Section   \ref{sec:discuss1}.
\begin{assumption}\label{ass:cons}
	The data are independent and identically distributed (iid).The function $$g(\be)=g_1(\bt)+g_2(\bp,\bt)=\lim_{n\rightarrow+\infty}\E[ n^{-1}\log g(\mathcal{D}|\be)]$$ exist and is such that: 1) for all $\delta>0$, $g_1(\bt_0)> \sup_{\|\bt-\bt_0\|>\delta}g_1(\bt)$; 2) for all $\delta>0$, $g_2(\bp_0,\bt_0)> \sup_{\|\bp-\bp_0\|>\delta}g_2(\bp,\bt_0)$.  The parameter space $\Theta\times\Phi$ is compact, and the function $\log g(\cdot\mid\be)$ is continuously differentiable at each $\be\in\Theta\times\Psi$ with probability one. For each $n\ge1$, $\E[\sup_{\be\in\Theta\times\Psi}\|\nabla_{\be} \log g(\mathcal{D}|\be)/n\|]<\infty$. 
\end{assumption}
To obtain the asymptotic distribution of the IFM we require some additional conditions on $m_n(\be):=n^{-1}(\{\nabla_{\be} \log g_1(\MD\mid\bt)\}^\top,\{\nabla_{\bp}\log g_2(\MD\mid\be)\}^\top)^\top$. 
\begin{assumption}\label{ass:dist2}
	The parameter $\be_0\in\mathrm{Int}(\Theta\times\Psi)$; there exists a $\delta>0$ such that, for all $\|\bt-\bt_0\|\le\delta$, $m_n(\be)$ is continuously differentiable in $\be$, and $\E[\sup_{\be:\|\be-\be_0\|\le\delta}\|\partial m_n(\be)/\partial\be^\top\|]<\infty$, and $\E[\|m_n(\be)\|^2]<\infty$. The matrix $\E[\partial m_n(\be)/\partial\be^\top]|_{\be=\be_0}$ is invertible.
\end{assumption}

\subsection{Discussion of Maintained Assumptions: Type~1 Cut Posterior}\label{sec:discuss1}

Assumption \ref{ass:cons}-\ref{ass:dist2} are standard regularity conditions employed in the analysis of frequentist point estimators obtained in two-steps. These conditions constitute a formalisation of the regularity conditions implicitly maintained in \cite{joe2005}, and which were used to deduce the asymptotic distribution of the IFM point estimator. 

Assumption \ref{ass:cons} is sufficient to deduce consistency of $\widehat\be$ for $\be_0$ using existing results on two-step estimators, see, e.g., Theorem 2.1 in \cite{newey1994large}. The first part of Assumption \ref{ass:cons} constitutes the usual identification condition required to deduce consistency of point estimators, while the second part of Assumption \ref{ass:cons} - in particular the differentiability and boundedness of the gradients - and compactness of $\Theta\times\Phi$ is sufficient to deduce a uniform law of large numbers that would ultimately drive consistent point estimation.

Assumption \ref{ass:cons} contains smoothness conditions on the gradient $m_n(\be)$ and are used to ensure that a Taylor-series expansion in a neighborhood of $\be_{0}$ can be applied, and that the remainder term in this expansion can be suitably controlled.  When used in conjunction with Assumption \ref{ass:cons}, Assumption \ref{ass:dist2} then allows us to obtain asymptotic normality of the IFM estimator $\sqrt{n}(\widehat\be-\be_0)$ using existing results, such as, e.g., Theorem 6.1 of \cite{newey1994large}. 

We recall that a key goal of Section \ref{sec:theory_cut1} is to demonstrate the theoretical equivalence between the IFM point estimator and the type~1 cut posterior mean. Consequently, the main goal of these assumptions is to help facilitate this comparison. As such, these conditions are far from the weakest set of assumptions under which the Bernstein-von Mises (BvM) results in Lemmas \ref{lem:two} and \ref{lem:ifm2} can be obtained. More generally, Lemmas  \ref{lem:two} and \ref{lem:ifm2} can be obtained without several components of Assumption \ref{ass:cons}, such as the compactness assumption, and the conditions that guarantee smoothness of the gradients.

In the context of the type~1 cut posterior, the compactness assumption can be weakened by assuming that the prior has sufficient regularity so that we can control the posterior mass outside of $\|\be-\be_{0}\|\ge \varepsilon$ for any $\varepsilon>0$ (see, e.g., Theorem 4 of \citealp{shen2001rates}). Twice continuous differentiability of the gradients can be weakened to continuity of $\log g(\MD\mid\be)$, and twice differentiability in a suitable neighborhood of $\be_0$. Alternatively, one may simply impose the following high-level regularity that is similar to that utilized in Theorem 2 of \cite{shen2001rates}, Theorem 2 in \cite{miller2021asymptotic}, and Theorem 1 in \cite{chernozhukov2003mcmc}: for any $\delta>0$, there exists $\varepsilon>0$ such that
\begin{equation}\label{eq:alt_smooth}
	\liminf_{n\rightarrow+\infty}P_0^{}\left[\sup_{\|\be-\be_0\|\ge\delta}\frac{1}{n}\left\{\log g(\MD\mid\be)-\log g(\MD\mid\be_0)\right\}\le-\varepsilon\right]=1.
\end{equation}The condition in \eqref{eq:alt_smooth}, along with conditions on the prior, are sufficient to control the behavior of the posterior outside of balls of radius $\|\be-\be_{0}\|\le\delta/\sqrt{n}$, for $\delta>0$. Given this condition, so long as the density $g(\MD\mid\be)$ is sufficiently smooth on $\|\be-\be_{0}\|\le\delta$ to permit a well-behaved quadratic expansion, we can dispense with the more stringent smooth conditions maintained in Assumption \ref{ass:cons}-\ref{ass:dist2}. 

Therefore, while it is possible to deduce the BvM results in  Lemmas  \ref{lem:two} and \ref{lem:ifm2} under conditions that are  weaker than Assumptions \ref{ass:cons}-\ref{ass:dist2}, such conditions are not representative of the conditions that are used to deduce the large sample behavior of the IFM estimator. Consequently, we have used the more stringer set of conditions rather than the weaker ones encountered when proving BvM results.

\subsection{Proofs of Results: Type~1 Cut Posterior}

Lemmas \ref{lem:ifm1}-\ref{lem:ifm2} in Section \ref{sec:theory_cut1} are a direct consequence of the following general result.
\begin{theorem}\label{thm:joint}
	If Assumptions \ref{ass:cons}-\ref{ass:dist2} in Appendix \ref{app:dtf} are satisfied, then
	$$
	\int_{}\|\be\|\left|p_{\cut}(\be|\MD)-\phi_{N}\{\be;\widehat\be,n^{-1}V\}\right| d  \be=o_p(1).
	$$ 
	
\end{theorem}
\begin{proof}[Proof of Theorem \ref{thm:joint}]
	Let $Q_n(\bt,\bp)=\log g_1(\MD\mid\bt)+\log g_2(\MD\mid\bp,\bt)$. Assumptions \ref{ass:cons}-\ref{ass:dist2} are sufficient to apply the result of Corollary 1 in \cite{frazier2022cutting} to $Q_n(\bt,\bp)$. Hence, for 
	$$
	t:=\sqrt{n}(\be-\be_0)-Z_n/\sqrt{n},\quad Z_n=V^{}\begin{pmatrix}
		\nabla_{\bt} \log g_1(\MD\mid \bt_0)\\\nabla_{\bp} \log g_2(\MD\mid \bp_0,\bt_0)
	\end{pmatrix}
	$$the posterior for $t$, $p_{\cut}(t\mid \MD)=p_{\cut}(\be_0+Z_n/\sqrt{n}+t/\sqrt{n}\mid\MD)/\sqrt{n}^{d_{\be}}$, satisfies
	\begin{flalign}\label{eq:result1}
		\int_{}\|t\||p_{\cut}(t\mid\MD)-\phi_{N}\{t;0,V^{}\}|dt=o_p(1).
	\end{flalign}
	Since $\|t\|\ge0$, \eqref{eq:result1} directly implies 
	\begin{flalign}\label{eq:result2}
		\int_{}|p_{\cut}(t\mid\MD)-\phi_{N}\{t;0,V^{}\}|dt=o_p(1).
	\end{flalign}
\end{proof}

\begin{proof}[Proof of {Lemmas \ref{lem:two}-\ref{lem:ifm2}.}]
	Since the total variation distance is invariant with respect to a change of location and/or scale, equation \eqref{eq:result2} immediately implies that 
	\begin{flalign*}
		\int_{}|p_{\cut}(\be\mid\MD)-\phi_{N}\{\be;\be_0+Z_n/n,n^{-1}V^{}\}|dt=o_p(1).
	\end{flalign*}
	However, noting that $Z_n/\sqrt{n}=\sqrt{n}(\widehat\be-\be_0)+o_p(1)$, which follows from \cite{joe2005}, we can re-express the above result as 
	\begin{flalign}\label{eq:result3}
		\int_{}|p_{\cut}(\be\mid\MD)-\phi_{N}\{\be;\widehat\be,n^{-1}V^{}\}|dt=o_p(1).
	\end{flalign}
	The results in Lemmas \ref{lem:two}-\ref{lem:ifm2} then follow directly from the structure of $V^{}$, and marginalization of Gaussian random variables. 
\end{proof}

\begin{proof}[Proof of {Lemma \ref{lem:ifm1}.}]
	To obtain the stated result in Lemma \ref{lem:ifm1}, first note that, from the change of variables $\be=\be_0+Z_n/n+t/\sqrt{n}$,
	$$
	\bar\be=\int \be p_{\cut}(\be\mid\MD)d\be=\int (\be_0+Z_n/n+t/\sqrt{n})p_{\cut}(t\mid\MD)dt,
	$$which yields
	$$
	\sqrt{n}(\bar\be-\be_0)-Z_n/\sqrt{n}=\int t p_{\cut}(t\mid\MD)dt=\int t[p_{\cut}(t\mid\MD)-\phi_{N}\{t;0,V^{}\}]dt.
	$$
	Bounding 
	$$
	\int t[p_{\cut}(t\mid\MD)-\phi_{N}\{t;0,V^{-1}\}]dt\le \int \|t\||p_{\cut}(t\mid\MD)-\phi_{N}\{t;0,V^{}\}|dt
	$$and using \eqref{eq:result1} then yields 
	$$
	o_p(1)=\|\sqrt{n}(\bar\be-\be_0)-Z_n/\sqrt{n}\|=\|\sqrt{n}(\bar\be-\widehat\be)\|.
	$$The stated result then follows since $Z_n/\sqrt{n}=\sqrt{n}(\widehat\be-\be_0)+o_p(1)$ under the maintained assumptions.	
\end{proof}
\newpage

\section{Assumptions and Results: Type~2 Cut Posterior}\label{app:rank}
To simplify the presentation we restrict the analysis to the case $m=2$. All results and derivations can be extended to $m\ge2$ at the cost of more cumbersome notations. Let $F_{1n},F_{2n}$ denote the re-scaled empirical CDF based on $(y_{11},y_{21}),\cdots,(y_{1n},y_{2n})$ as $$F_{jn}(y_{j})=\frac{r(y_{j})}{(1+n)}=\frac{1}{1+n}\sum_{k=1}^{n}\mathds{1}(y_{jk}\le y_j),$$ where $r(y_j)=\sum_{k=1}^{n}\mathds{1}(y_{jk}\le y_j)$ is the rank of $y_j$. Let
$$
C_n(u_1,u_2)=\frac{1}{1+n}\sum_{k=1}^{n}\mathds{1}[F_{1n}(y_{1k})\le u_1,F_{2n}(y_{2k})\le u_2]
$$ denote the re-scaled empirical copula process. Let $U_k=(U_{1,k},U_{2,k})^\top$ with $U_{j,k}=F_{j0}(y_j)$, where, for $j=1,2$, $F_{j0}(y_j)$ denotes the true CDF of $y_j$. For some unknown copula function 
$
C_0(u_1,u_2)
$ we have
$$
P^{}_0(y_{1,k},y_{2,k})=C_0(U_{1,k},U_{2,k}).
$$ 

Our preliminary interest is in the asymptotic behavior of statistics of the form
$$
R_n(\bp)=\frac{1}{1+n}\sum_{k=1}^{n}J\{F_{1n}(y_{1k}),F_{2n}(y_{2k});\bp\}=\int J(u_1,u_2;\bp)d C_n(u_1,u_2),
$$ for some $u_1,u_2\in[0,1]^2$,   and $J:[0,1]^2\times \Psi\rightarrow\mathcal{J}$ continuously differentiable in both arguments. Recall the logarithm of the pseudo rank likelihood at \eqref{eq:rlike} in the main text:
\begin{flalign*}
	M_n(\bp)&=\sum_{k=1}^{n} \log\{\Delta^{b_{k1}}_{a_{k1}}\Delta^{b_{k2}}_{a_{k2}}C(v;\bp)\}\\&= \sum_{k=1}^{n}\log \left[\left\{C(b_{1k},b_{2k})-C(b_{1k},a_{2k})\right\}-\left\{C(a_{1k},b_{2k})-C(a_{1k},a_{2k})\right\}\right],
\end{flalign*}
where we recall that  
$$
b_{jk}=\frac{r(y_{jk})}{(1+n)}=F_{jn}(y_{jk}),\text{ and } a_{jk}=b_{jk}-\frac{1}{(1+n)}.
$$Hence, for $u_1,u_1',u_2,u_2'\in(0,1)^4$, and $u_1',u_2'$ such that $u_1'=u_1-(1+n)^{-1}$, $u_2'=u_2-(1+n)^{-1}$, letting
\begin{equation}\label{eq:newJ}
	J\{u_1,u_2;\bp\}=\log\left[ \{C(u_1,u_2;\bp)-C(u_1,u_2';\bp)\}-\{C(u_1',u_2;\bp)-C(u_1',u_2';\bp)\}\right],	
\end{equation}
we can restate the rank log-likelihood as 
$$
M_n(\bp)=(1+n)\int J(u_1,u_2;\psi)d C_n(u_1,u_2).
$$So long as $J(\cdot,\cdot;\bp)$ is twice continuously differentiable in $\bp$, we can obtain the behavior of ${M}_{n,\bp}(\bp)=\nabla_{\bp}M_n(\bp)/(1+n)$ and ${M}_{n,\bp,\bp}(\bp)=\nabla^2_{\bp,\bp}M_n(\bp)/(1+n)$ using the behavior of the copula process $C_n(u_1,u_2)$. Lastly, recall the definition
$
\mathcal{M}(\bp)=\lim_{n\rightarrow+\infty}\E M_n(\bp)/(1+n).
$

\subsection{Maintained Assumptions: Type~2 Cut Posterior}

To analyze the behavior of statistics like $R_n(\bp)$, we impose the following regularity conditions. 

\begin{assumption}\label{ass:DGP1}(i) $Y_1=\left(Y_{11}, \ldots, Y_{m 1}\right)^{\top}, \ldots, Y_n=\left(Y_{1 n}, \ldots, Y_{m n}\right)^{\top}$ is an independent and identically distributed sample from the unknown distribution $F_0\left(y_1, \ldots, y_m\right)$ with continuous marginal distributions $F_{10}, \ldots, F_{m0}$. (ii) The true (unknown) copula function $C_0\left(u_1, \ldots, u_m\right)$ has continuous partial derivatives. 
	
\end{assumption}	

\begin{assumption}\label{ass:Jfun}
	$J:[0,1]^m \rightarrow(-\infty, \infty)$ is a continuous function having continuous partial derivatives $J_j(u)=\partial J(u) / \partial u_j$ on $(0,1)^m$ for $j=1, \ldots, m$. Suppose:
	(i) $|J(u)| \leq$ constant $\times \prod_{j=1}^m\left\{u_j\left(1-u_j\right)\right\}^{-a_j}$ for some $a_j \geq 0$ such that
	$$
	\mathrm{E}\left[\prod_{j=1}^m\left\{U_{j t}\left(1-U_{j t}\right)\right\}^{-2 a_j}\right]<\infty;
	$$(ii) $\left|J_k(u)\right| \leq$ constant $\times\left\{u_k\left(1-u_k\right)\right\}^{-b_k} \prod_{j=1, j \neq k}^m\left\{u_j\left(1-u_j\right)\right\}^{-a_j}$ for some $b_k>a_k$ such that
	$$
	\mathrm{E}\left[w_k\left(U_{k t}\right)\left\{U_{k t}\left(1-U_{k t}\right)\right\}^{-b_k} \prod_{j=1, j \neq k}^m\left\{U_{j t}\left(1-U_{j t}\right)\right\}^{-a_j}\right]<\infty
	$$
	for $k=1, \cdots, m$.
\end{assumption}
\begin{assumption}\label{ass:crit1} (i) For any $u \in(0,1)^m$ and $J(u;\bp)$ defined in \eqref{eq:newJ}, $J(u;\bp)$ is a twice-continuously differentiable function of $\bp$; (ii) $\mathrm{E}\left\{\sup _{\bp \in \Phi}\left|J(U_{k};\bp)\right|\right\}<\infty$. (iii) $\mathrm{E}\left\{J(U_{k};\bp)\right\}$ has a unique maximum $\bp_\star$ in $\Phi$, where $\Phi$ is a compact subset of $\mathbb{R}^{d_{\bp}}$; (iv) the map $\bp\mapsto\mathcal{M}_{\bp\bp}(\bp)$ is positive definite at $\bp=\bp_\star$.	
\end{assumption}

\subsection{Discussion of Maintained Assumptions: Type~2 Cut Posterior}\label{sec:discuss2}
Herein, we provide a brief discussion of Assumptions \ref{ass:DGP1}-\ref{ass:crit1}, and suggest when it may be possible to weaken them in certain examples. Similar to Assumptions \ref{ass:cons}-\ref{ass:dist2},  Assumptions \ref{ass:DGP1}-\ref{ass:crit1} ensure that the posterior based on the pseudo rank likelihood concentrates sufficiently fast, and has enough smoothness to permit a quadratic approximation around $\bp_{\star}$. 

Assumption \ref{ass:DGP1} and \ref{ass:crit1} are common assumptions employed in the study of copula processes, and ensures that the criterion function $M_n(\bp)$ and its derivatives are well-behaved (see, e.g., \citealp{genest95} for similar assumptions). Together, these assumptions enforce concentration of the type~2 cut posterior for $\bp$. Assumption \ref{ass:Jfun} and \ref{ass:crit1} together jointly enforce asymptotic normality of the type~2 cut posterior in a neighborhood of $\bp_\star$. Assumptions like \ref{ass:Jfun} are well-known in the literature on copula modeling, and ensures that the derivatives of $J(u;\bp)$ are allowed to explode at the boundaries of $[0, 1]^m$, but that this rate of explosion can be controlled. Similar condition can be found in Assumption A3 of \cite{chen2005pseudo}, and Condition 12 in \cite{alquier2023estimation}. 

To our knowledge, Assumptions \ref{ass:DGP1}--\ref{ass:crit1}  have been widely adopted in copula estimation based on general criterion functions. For example, \cite{alquier2023estimation} employ similar conditions to our Assumptions \ref{ass:DGP1}--\ref{ass:crit1} when proving consistency and asymptotic normality of copula dependence parameter estimates obtained by minimizing the maximum mean discrepancy divergence. While it may be possible to weaken Assumptions \ref{ass:DGP1}-\ref{ass:crit1} using approaches similar to those discussed in Section \ref{sec:discuss1}, the main goal of this analysis is not to present the weakest set of regularity conditions under which our results will hold, but to demonstrate the usefulness of cutting feedback methods in copula modeling. Hence, we leave an exploration of alternative assumptions under which the results in Section \ref{sec:theory_cut2} will be satisfied for further research. 

We note that in many copula models Assumptions \ref{ass:DGP1}-\ref{ass:crit1} are satisfied under moment assumptions on the true data generating process. Assumption \ref{ass:Jfun} is specifically tailored to allow for partial derivatives of $J(\cdot,\cdot;\bp)$ to explode at the boundaries of $[0,1]^m$, but not ``too quickly''. This condition is known to be satisfied for certain elliptical copulas, like the student-t or Gaussian copula. 

In the case of elliptical copulas, like the Gaussian, the compactness assumption on the parameter space $\Phi$ in Assumption \ref{ass:crit1} is satisfied. For Archimedean copulas, such as the Clayton or Gumbel copulas, the dependence parameters are not compactly supported, as these copulas can exhibit perfect dependence between the random variables in the limit. For example, in the case of the bivariate Gumbel copula, 
$\exp \left[-\left\{(-\log (u))^\alpha+(-\log (v))^\alpha\right\}^{1 / \alpha}\right]$ with $1\le\alpha<\infty$, when $\alpha$ is large the copula  exhibits stronger dependence and becomes arbitrarily close to perfectly dependent when $\alpha$ is very large; a similar argument to the above is also true for the Clayton copula.  However, if the random variables are not jointly singular, then they must exhibit finite dependence; i.e., in the Gumbel copula model it must be that $\alpha\le M<\infty$ for some $M$ arbitrarily large but finite. This means that in any realistic empirical example, it is almost without loss of generality to consider a compact parameter space, such as $[1,M]$  in the case of the Gumbel copula, with a  similar argument also being valid for other Archimedean copulas.  Consequently, the compactness assumption, while satisfied for certain copulas, and not for others, is really a technical regularity condition that one needs to deliver theoretical guarantees, but which does not impose any meaningful restriction on the types of behaviors that can be accurately modeled, and which are observed in practice.

\subsection{Proofs of Results: Type~2 Cut Posterior}

\subsubsection{Preliminary Lemmas}
Write $G_{C_0}$ to be a Gaussian process in $\ell^\infty([0,1]^m)$ defined as 
$$
G_{C_0}\left(u_1, \ldots, u_m\right)=B_{C_0}\left(u_1, \ldots, u_m\right)-\sum_{j=1}^m \frac{\partial C_0\left(u_1, \ldots, u_m\right)}{\partial u_j} B_{C_0}\left(1, \ldots, 1, u_j, 1, \ldots, 1\right),
$$
in which $B_{C_{0}}$ is a Brownian bridge on $[0,1]^m$ with covariance function
$$
\begin{aligned}
	\mathrm{E} & \left\{B_{C_0}\left(u_1, \ldots, u_m\right) B_{C_0}\left(u_1^{'}, \ldots, u_m^{'}\right)\right\}  =C_0\left(u_1 \wedge u_1^{'}, \ldots, u_m \wedge u_m^{'}\right)-C_0\left(u_1, \ldots, u_m\right) C_0\left(u_1^{'}, \ldots, u_m^{'}\right)
\end{aligned}
$$
for each $0 \leq u_1, \ldots, u_m, u_1^{'}, \ldots, u_m^{'} \leq 1$. 

The following lemmas are given in \cite{chen2005pseudo}, and are used to prove our main results. They are restated here for ease of reference. 
\begin{lemma}[Lemma 1 \cite{chen2005pseudo}]\label{lem:CF1}  Under the conditions in Assumption \ref{ass:DGP1}, the following are satisfied. 
	(a) The re-scaled empirical copula process $\left\{\sqrt{n}\left\{{C}_n(u)-C_0(u)\right\}: [0,1]^m\right\}$ converges weakly to the Gaussian process $\left\{G_{C_0}(u): [0,1]^m\right\}$ in $\ell^{\infty}\left([0,1]^m\right)$.
	(b) Let $\mathcal{H}$ be a class of functions $h:[0,1]^m \rightarrow(-\infty, \infty)$ which satisfies: for every $\delta>0$, $N_{[\cdot]}\left(\delta, \mathcal{H}, L_1\left(C_0\right)\right)<\infty$, where $N_{[\cdot]}\left(\delta, \mathcal{H}, L_1\left(C_0\right)\right)$ is the $L_1\left(C_0\right)$ bracketing number of the class $\mathcal{H}$. Then:
	$$
	\sup _{h \in \mathcal{H}}\left|\int_{[0,1]^m} h(u) d\left\{{C}_n(u)-C_0(u)\right\}\right| \rightarrow 0 \text { almost surely. }
	$$
\end{lemma}

\begin{lemma}[Lemma 2 \cite{chen2005pseudo}]\label{lem:CF2}
	Suppose Assumption \ref{ass:DGP1} and \ref{ass:Jfun} hold. For $j=1, \ldots, m$, let $w_j$ in Assumption \ref{ass:Jfun} be a continuous function on $[0,1]$, positive on $(0,1)$, symmetric at $1 / 2$, increasing on $(0,1 / 2]$ and such that
	$$
	\int_0^1\left\{\frac{1}{w_j(v)}\right\}^2 d v<\infty
	$$
	Then
	$$
	\sqrt{n} \int_{[0,1]^m} J(u) d\left\{{C}_n(u)-C_0(u)\right\} \rightarrow N\left(0, \sigma_J^2\right) \text { in distribution }
	$$
	where
	$$
	\sigma_J^2=\operatorname{var}\left\{J\left(U_{1 t}, \ldots, U_{m t}\right)+\sum_{j=1}^m \int_{[0,1]^m} J_j(u) \times I\left(U_{j t} \leq u_j\right) d C_0(u)\right\}.
	$$
\end{lemma}

By applying Lemmas \ref{lem:CF1} and \ref{lem:CF2} we can prove the following results.

\begin{corollary}\label{cor:two}
	If Assumptions \ref{ass:DGP1}-\ref{ass:crit1} are  satisfied, then  the following hold. 
	\begin{enumerate}
		\item[(1)] $\sup_{\bp\in\Psi}|M_n(\bp)/(1+n)-\mathcal{M}_{}(\bp)|=o_p(1)$.
		\item[(2)] For any $\delta_n=o(1)$, and ,  $\sup_{\|\bp-\bp_\star\|\le\delta_n}\|\nabla^2_{\bp\bp}M_{n}(\bp)/(1+n)-\nabla^2_{\bp\bp}\mathcal{M}(\bp)\|=o_p(1)$.
		\item[(3)] $\sqrt{n}\nabla_{\bp} M_n(\bp_\star)/(1+n)\Rightarrow N(0,\Sigma)$, for some positive definite matrix 
		$$
		\Sigma=\operatorname{var}\left\{\tilde{J}\left(U_{1 t}, U_{2 t}\right)+\sum_{j=1}^2 \int_{[0,1]^2} \nabla_{u_j}\tilde{J}(u_1,u_2) \times I\left(U_{j t} \leq u_j\right) d C_0(u_1,u_2)\right\},
		$$ and $\tilde{J}(U_{1t},U_{2t})=\nabla_{\bp} J(u_1,u_2;\bp_\star)$, with $J(u_1,u_2;\bp_\star)$ defined in \eqref{eq:newJ}.
	\end{enumerate}		
\end{corollary}

\begin{proof}[Proof of Corollary \ref{cor:two}] Recall that we take $m=2$ for exposition purposes, and again note that this can be extended to $m>2$ at the cost of more cumbersome notations. We prove each result in turn. 
	
	\noindent \textbf{(1)} Note, $(1+n)^{-1}M_n(\bp)=\int_0^1J(u;\bp)dC_n(u)$, and $\mathcal{M}(\bp)=\int_{[0,1]^2} J(u;\bp)dC_0(u)$, for $J(u;\bp)$ as defined in \eqref{eq:newJ}. Since $\Phi$ is compact, by the continuity of $M_n(\bp)$ in Assumption \ref{ass:crit1}(i), the class of functions $\mathcal{H}:=\{J(u;\bp):\bp\in\Phi\}$ is such that $N_{[\cdot]}(\delta,\mathcal{H},L_1(C_0))<\infty$. Hence, Lemma \ref{lem:CF1}(b) implies the stated result.  
	\\
	
	\noindent\textbf{(2)} Now, write $\nabla^2_{\bp\bp}M_n(\bp_\star)/(1+n)=\int_{[0,1]^2} h(u;\bp)dC_n(u)$, and $\nabla^2_{\bp\bp}\mathcal{M}(\bp_\star)=\int_{[0,1]^2}h(u;\bp_\star) dC_n(u) $, for $h(u;\bp)=\nabla^2_{\bp\bp}J(u;\bp)$, which is continuous by Assumption \ref{ass:crit1}(1). For any $\delta_n=o(1)$, the class of functions $\mathcal{H}:=\{\nabla_{\bp\bp}^2J(u;\bp):\|\bp-\bp_\star\|\le\delta_n\}$ satisfies $N_{[\cdot]}(\delta,\mathcal{H},L_1(C_0))<\infty$. Hence, Lemma \ref{lem:CF1}(b) implies the stated result.
	\\
	
	\noindent\textbf{(3)} Rewriting $\sqrt{n}\nabla_{\bp}M_n(\bp_\star)/(1+n)=\sqrt{n}\int_{[0,1]^2}\{\nabla_{\bp}J(u;\bp_\star)\}d\{C_n(u)-C_0(u)\}$, the stated result follows directly by Lemma \ref{lem:CF2}. 	
\end{proof}	

\subsubsection{Proofs of Results: Type~2 Cut Posterior}

\begin{proof}[Proof of Theorem \ref{thm:ranks}]
	Recall the pseudo rank log-likelihood given by $M_n(\bp)$, and rewrite the cut posterior for $\bp$ as 
	$$
	p_{\cut}(\bp\mid\mathcal{D})=\frac{p_{\text{PL}}(r(\mathcal{D})\mid\bp)p(\bp)}{\int_{\Psi}p_{\text{PL}}(r(\mathcal{D}\mid\bp)p(\bp)d\bp}=\frac{\exp\{M_n(\bp)-M_n(\bp_\star)\}p(\bp)}{\int_{\Psi}\exp\{M_n(\bp)-M_n(\bp_\star)\}p(\bp)d\bp}.
	$$	For any  $\delta=o(1)$, split $\Phi$ into $\Phi=\Phi_\delta\cup\Phi_\delta^c$, where $\Phi_\delta:=\{\bp\in\Phi:\|\bp-\bp_\star\|\le\delta\}$.  
	\\
	
	\noindent\textbf{Region $\Phi_\delta^c$:}
	From Corollary \ref{cor:two}(1), we have that 
	\begin{equation}\label{eq:lln}
		\sup_{\bp\in\Phi}|(1+n)^{-1}M_n(\bp)-\mathcal{M}(\bp)|=o_p(1).	
	\end{equation}
	From equation \eqref{eq:lln}, for any $\delta>0$,
	\begin{flalign*}
		\sup_{\|\bp-\bp_\star\|\ge \delta}(1+n)^{-1}\left\{M_n(\bp)-M_n(\bp_\star)\right\}\leq& \sup_{\|\bp-\bp_\star\|\ge \delta}2|(1+n)^{-1}M_n(\bp)-\mathcal{M}(\bp)|\\&+\sup_{\|\bp-\bp_\star\|\ge \delta}\left\{\mathcal{M}(\bp)-\mathcal{M}(\bp_\star)\right\}.
	\end{flalign*}
	From equation \eqref{eq:lln} the first term is $o_p(1)$, and from Assumption \ref{ass:crit1}(iii), for any $\delta>0$ there exists an $\epsilon>0$ such that
	$
	\sup_{\|\bp-\bp_\star\|\ge \delta}\left\{\mathcal{M}(\bp)-\mathcal{M}(\bp_\star)\right\}\le -\epsilon.
	$
	Hence, 
	\begin{equation}
		\label{eq:expconv}
		\lim_{n\rightarrow+\infty}P^{}_0\left[\sup_{\|\bp-\bp_\star\|\geq \delta}\exp\left\{M_n(\bp)-M_n(\bp_\star)\right\}\leq \exp\{-\epsilon (1+n)\}\right]=1.
	\end{equation}
	Using the above, over $\Phi_\delta^c$, 
	\begin{align*}
		\int_{\|\bp-\bp_{0}\|\ge \delta }p\left(\bp\right)\exp\left\{M_n(\bp)-M_n(\bp_\star) \right\}d \bp&\leq \{1+o_p(1)\}O_{P_0}(1) \exp\left\{-n\epsilon\right\}\int_{\|\bp-\bp_0\|\ge \delta }p\left(\bp\right)d \bp\\&\leq \{1+o_p(1)\}O_{P_0}(1) \exp\left\{-n\epsilon\right\}
		\\&=o_p(1). 
	\end{align*} 
	The above directly implies that
	\begin{flalign*}					p_{\cut}(\bp\mid\MD)&=\frac{\exp\left\{M_n(\bp)-M_n(\bp_\star)\right\}p(\bp)}{\int_\Phi\exp\left\{M_n(\bp)-M_n(\bp_\star)\right\}p(\bp)d \bp}=\frac{\exp\left\{M_n(\bp)-M_n(\bp_\star)\right\}p(\bp)}{\int_{\Phi_\delta}\exp\left\{M_n(\bp)-M_n(\bp_\star)\right\}p(\bp)d \bp+o_p(1)}.
	\end{flalign*} Integrating over $\Phi_\delta$ demonstrates that, for any $\delta>0$,  as $n\rightarrow+\infty$, 
	$$
	\int_{\Phi_\delta}p_{\cut}(\bp\mid\MD)d\bp=1+o_p(1). 
	$$

	\noindent\textbf{Region $\Phi_\delta$.} Consider the second-order Taylor series expansion of $M_n(\bp)-M_n(\bp_\star)$:  for some line-by-line intermediate value $\bar\bp$, such that $\|\bp_\star-\bar\bp\|\le\|\bp-\bp_\star\|$,  
	\begin{flalign*}
		M_n(\bp)-M_n(\bp_\star)	&=(\bp-\bp_\star)'\nabla_{\bp} M_n(\bp)+\frac{1}{2}(\bp-\bp_\star)'{\nabla_{\bp\bp}^2 M_n(\bar{\bp})}{}(\bp-\bp_\star)\\&=\sqrt{(1+n)}(\bp-\bp_\star)'\mathcal{M}_{\bp\bp}Z_n/\sqrt{(1+n)}-\frac{1}{2}\sqrt{(1+n)}(\bp-\bp_\star)'\mathcal{M}_{\bp\bp}\sqrt{(1+n)}(\bp-\bp_\star)\\&+\frac{1}{2}\sqrt{(1+n)}(\bp-\bp_\star)\left\{\frac{1}{1+n}{\nabla_{\bp\bp} M_n(\bar{\bp})}-{\nabla_{\bp\bp} \mathcal{M}(\bar{\bp})}{}\right\}\sqrt{(1+n)}(\bp-\bp_\star)
		\\&-\frac{1}{2}\sqrt{(1+n)}(\bp-\bp_\star)\left\{{\nabla_{\bp\bp} \mathcal{M}(\bp)}-{\nabla_{\bp\bp} \mathcal{M}(\bar{\bp})}\right\}\sqrt{(1+n)}(\bp-\bp_\star)
	\end{flalign*}	where we recall that $\mathcal{M}_{\bp\bp}=-\nabla_{\bp\bp} \mathcal{M}(\bp_\star)$, and let $Z_n=\mathcal{M}_{\bp\bp}^{-1}(1+n)\int_0^1\nabla_{\bp}J(u;\bp_\star)d\{C_n(u)-C_0(u)\}$.  Represent the last two terms in the above equation as $R_n(\bp)=R_{1n}(\bp)+R_{2n}(\bp)$, and re-arrange the first two terms to obtain 
	\begin{flalign}\label{eq:tse1}
		M_n(\bp)-M_n(\bp_\star)	&=-\frac{1}{2}t'\mathcal{M}_{\bp\bp}t+\frac{1}{2n}Z_n'\mathcal{M}^{}_{\bp\bp}Z_n+R_{1n}(\bp)+R_{2n}(\bp), 
	\end{flalign}where $t:=\sqrt{(1+n)}(\bp-\bp_\star)-Z_n/\sqrt{(1+n)}$. 
	
	For $T_n=\bp_\star+Z_n/\sqrt{(1+n)}$, use the change of variables $\bp\mapsto t$ to obtain 
	$$
	p_{\cut}(t\mid\MD) =\frac{p_{\cut}(t/\sqrt{(1+n)}+T_n\mid\MD)}{n^{d_{\bp}/2}}=\frac{\exp\left\{\xi_n(t)\right\}p(t/\sqrt{(1+n)}+T_n)}{\int_{\mathcal{T}}\exp\left\{\xi_n(t)\right\}p(t/\sqrt{(1+n)}+T_n)dt},
	$$
	where $\xi_n(t):=M_n(t/\sqrt{(1+n)}+T_n)-M_n(\bp_\star)-\frac{1}{2}\frac{1}{n}Z_n'\mathcal{M}_{\bp\bp}Z_n$, and $\mathcal{T} :=\{t=\sqrt{(1+n)}(\bp-\bp_\star)-Z_n/\sqrt{(1+n)}:\bp\in\Phi_\delta\}$. Applying \eqref{eq:tse1}, into the definition of $\xi_n(t)$, then yields the following form for the cut posterior:
	\begin{equation}\label{eq:xit}
		p_{\cut}(t\mid\MD)	=\frac{\exp\{-\frac{1}{2}t'\mathcal{M}_{\bp\bp}t+R_{1n}(t/\sqrt{(1+n)}+T_n)+R_{2n}(t/\sqrt{(1+n)}+T_n)\}p(t/\sqrt{(1+n)}+T_n)}{\int_{\mathcal{T}}\exp\{-\frac{1}{2}t'\mathcal{M}_{\bp\bp}t+R_{1n}(t/\sqrt{(1+n)}+T_n)+R_{2n}(t/\sqrt{(1+n)}+T_n)\}p(t/\sqrt{(1+n)}+T_n)d t}.	
	\end{equation}
	If we can show that the remainder terms $R_{1n}(\bp)$, $R_{2n}(\bp)$, satisfy 
	$$
	\sup_{\|\bp-\bp_\star\|\le\delta}\frac{R_{1n}(\bp)}{(1+n\|\bp-\bp_\star\|^2)}=o_p(1),\;\sup_{\|\bp-\bp_\star\|\le\delta}\frac{R_{2n}(\bp)}{(1+n\|\bp-\bp_\star\|^2)}=o_p(1),
	$$ then the result follows using arguments similar to those in Theorem 1 of \cite{frazier2022cutting}.
	
	To show that $\sup_{\|\bp-\bp_\star\|\le\delta}\frac{R_{1n}(\bp)}{(1+n\|\bp-\bp_\star\|)}$, we note that 
	\begin{flalign*}
		&\sup_{\|\bp-\bp_\star\|\le\delta}\frac{R_{1n}(\bp)}{(1+n\|\bp-\bp_\star\|^2)}\\&\le \sup_{\|\bp-\bp_\star\|\le\delta}\left\|\left\{\frac{1}{n}{\nabla_{\bp\bp} M_n({\bp})}-{\nabla_{\bp\bp} \mathcal{M}({\bp})}{}\right\}\right\|^2\frac{\|\sqrt{(1+n)}(\bp-\bp_\star)\|^2}{1+\|\sqrt{(1+n)}(\bp-\bp_\star)\|^2}\\&\le \sup_{\|\bp-\bp_\star\|\le\delta}\left\|\left\{\frac{1}{n}{\nabla_{\bp\bp} M_n({\bp})}-{\nabla_{\bp\bp} \mathcal{M}({\bp})}{}\right\}\right\|^2=o_p(1),
	\end{flalign*}where the last line follows from Corollary \ref{cor:two}(2). 
	To demonstrate that 
	$\sup_{\|\bp-\bp_\star\|\le\delta}\frac{R_{2n}(\bp)}{(1+n\|\bp-\bp_\star\|)}$, we note that 
	\begin{flalign*}
		&\sup_{\|\bp-\bp_\star\|\le\delta}\frac{R_{2n}(\bp)}{(1+n\|\bp-\bp_\star\|^2)}\\&\le \sup_{\|\bp-\bar\bp_\star\|\le\delta}\left\|\left\{ {\nabla_{\bp\bp} \mathcal{M}({\bp})}-{\nabla_{\bp\bp} \mathcal{M}(\bar{\bp})}{}\right\}\right\|^2\frac{\|\sqrt{(1+n)}(\bp-\bp_\star)\|^2}{1+\|\sqrt{(1+n)}(\bp-\bp_\star)\|^2}\\&\le \sup_{\|\bp-\bp_\star\|\le\delta}\left\|\left\{{\nabla_{\bp\bp} \mathcal{M}({\bp})}-{\nabla_{\bp\bp} \mathcal{M}(\bar{\bp})}{}\right\}\right\|^2. 
	\end{flalign*}Recall that $\bar\bp$ is a line-by-line intermediate value such that $\|\bp_\star-\bar\bp\|\le\|\bp-\bp_\star\|$. Hence, for any $\delta=o(1)$, and $\|\bp-\bp_\star\|\le\delta$, continuity of $\nabla_{\bp\bp}\mathcal{M}(\bp)$ implies that 
	$$
	\sup_{\|\bp-\bp_\star\|\le\delta}\left\|\left\{{\nabla_{\bp\bp} \mathcal{M}({\bp})}-{\nabla_{\bp\bp} \mathcal{M}(\bar{\bp})}{}\right\}\right\|^2=o(1). 
	$$

\end{proof}

\begin{proof}[Proof of Theorem \ref{thm:cut2}]
	The result follows since Assumptions  \ref{ass:DGP1}-\ref{ass:crit1} are sufficient to apply Corollary 1 of \cite{frazier2022cutting} to $Q_n(\bp,\bt)=M_n(\bp)+\log p(\MD\mid\bp,\bt)$.
\end{proof}

\newpage
\section{Additional results for Simulation 1}
In this appendix we provide additional results for Simulation~1 when the sample size is $n=100$ and $n=500$. The priors are given in
the text, although we note here that for the Gumbel copula in the misspecified models the prior on Kendall's $\tau$ is uniform on the admissible region $[0,1)$ and 
for the t-copula it is uniform on the admissible region $(-1,1)$.

\begin{table}[htbp]
	\caption{Parameter Point Estimation Accuracy in Simulation 1 ($n=100$)} 
	\label{tab:sim1biasrmseN100}    
	\begin{center}
		\begin{tabular}{ccccccccc}
			\hline\hline
			& \multicolumn{5}{c}{Misspecified Copula Fit} & &\multicolumn{2}{c}{Correct Copula Fit} \\ \cline{2-6}\cline{8-9} 
			& Uncut/ & Cut/ & IFM & Uncut/ & Cut/ & &Uncut/ &Cut/\\
			& MCMC & MCMC & &VI &VI & &MCMC &MCMC\\ \hline
			Parameter &{\em Bias} & & & & & & &\\ \cline{2-6}\cline{8-9} 
			$\mu$ & 0.0071 & 0.0041 &  0.0042 & {\bf 0.0029} & 0.0031 & &0.0019 &0.0041\\ 
			$\sigma^2$ & 0.0856 & 0.0283 &  {\bf 0.0015} & 0.0998 & 0.0986  & &0.0015 &0.0283\\ 
			$\alpha$ & -0.2685 & -0.0557 & 0.1560  & -0.2071 & {\bf -0.0150} & &-0.0101 & -0.0557\\
			$\beta$ & -0.1289 & -0.0260 & 0.0659  & -0.0968 & {\bf -0.0060} & &-0.0046 &-0.0260\\
			$\tau$ & 0.0328 & -0.0266 & {\bf 0.0179}  & 0.0329 &  - 0.0327 & &-0.0034 &-0.0078\\
			&{\em RMSE} & & & & & & &\\ \cline{2-6}\cline{8-9}
			$\mu$ & 0.0899 & 0.0902 & 0.0902  & {\bf 0.0900} & 0.0905 & & 0.0823 & 0.0902\\
			$\sigma^2$ & 0.1683 & 0.1432 &  {\bf 0.0683}  & 0.1764 & 0.1791 & & 0.1408 & 0.1432\\
			$\alpha$ & 0.9682 & {\bf 0.9309} & 0.9957  & 0.9560 & 0.9461 & & 0.8389 & 0.9309\\
			$\beta$ & 0.4273 & {\bf 0.4127} & 0.4410  & 0.4201 & 0.4191 & & 0.3740 & 0.4127\\
			$\tau$ & 0.0636 & 0.0555 & {\bf 0.0532}  & 0.0633 & 0.0584 & & 0.0412 & 0.0486\\ \hline\hline
		\end{tabular}
	\end{center}
	The bias and RMSE values of different point estimators computed over the $S=500$ simulation replicates, with the lowest values
	in bold. Results on the left are where the misspecified copula is fit using the posterior mean from the conventional (i.e. uncut) and cut (type 1) posteriors, computed exactly using MCMC or approximately using VI. IFM is included for comparison. Results on the right are where the correct copula is fit using the 
	posterior mean from the conventional (i.e. uncut) and cut (type 1) posteriors
	computed exactly using MCMC.
\end{table}

\begin{table}[htbp]
	\caption{Coverage and Copula Model Accuracy in Simulation 1 ($n=100$)} \label{tab:sim1covklN100}
	\begin{center}
		\renewcommand\arraystretch{1.15}
		\begin{tabular}{ccccccccc}
			\hline\hline
			& \multicolumn{5}{c}{Misspecified Copula Fit} & &\multicolumn{2}{c}{Correct Copula Fit} \\ \cline{2-6}\cline{8-9} 
			& Uncut/ & Cut/ & IFM & Uncut/ & Cut/ & &Uncut/ &Cut/\\
			& MCMC & MCMC & &VI &VI & &MCMC &MCMC\\ \hline
			Parameter &\multicolumn{3}{l}{\em Coverage Probabilities} & & & & &\\ \cline{2-6}\cline{8-9} 
			$\mu$ & 0.9740 & {\bf 0.9600} & - & 0.9760 & 0.9720 & &0.9520 &0.9600\\ 
			$\sigma^2$ & 0.8380 & {\bf 0.9600} & - & 0.9760 & 0.9640 & &0.9500 & 0.9600\\
			$\alpha$ & 0.8300 & {\bf 0.9420} & - & 0.8740 & 0.9740 & &0.9620 &0.9420\\
			$\beta$ & 0.9320 & {\bf 0.9360} & - & 0.8760 & 0.9760 & &0.9520 &0.9360\\
			$\tau$ & 0.6560 & {\bf 0.8740} & - & 0.6400 & 0.8680 & &0.9440 &0.8920\\ 
			Component &\multicolumn{3}{l}{\em Mean Predictive  KL Divergence} & & & & &\\ \cline{2-6}\cline{8-9} 
			$f_1$ & 0.0102 & 0.0084 &  {\bf 0.0057} & 0.0108 & 0.0115 & & 0.0079 & 0.0084\\
			$f_2$ & 0.0096 & 0.0091 & 0.0095 & 0.0092 & {\bf 0.0090} & & 0.0074 & 0.0091\\ 
			$c$ & 0.1795 & 0.1546 & {\bf 0.1510} & 0.1793 & 0.1659 & & 0.0074 & 0.0235\\ \hline\hline
		\end{tabular}
	\end{center}
	Top: coverage probabilities for 95\% credible intervals for each parameter, with
	those closest to 0.95 in bold. Bottom: the mean predictive KL divergence for the two marginals and the copula density for their estimate, with the lowest values 
	in bold. 
	Results on the left are where the misspecified copula is fit using the conventional (i.e. uncut) and cut (type 1) posteriors, computed exactly using MCMC or approximately using VI. IFM is included for comparison. Results on the right are where the correct copula is fit using the conventional (i.e. uncut) and cut (type 1) posteriors
	computed exactly using MCMC.
\end{table}

\begin{table}[htbp]
	\caption{Parameter Point Estimation Accuracy in Simulation 1 ($n=500$)} 
	\label{tab:sim1biasrmseN500}    
	\begin{center}
		\begin{tabular}{ccccccccc}
			\hline\hline
			& \multicolumn{5}{c}{Misspecified Copula Fit} & &\multicolumn{2}{c}{Correct Copula Fit} \\ \cline{2-6}\cline{8-9} 
			& Uncut/ & Cut/ & IFM & Uncut/ & Cut/ & &Uncut/ &Cut/\\
			& MCMC & MCMC & &VI &VI & &MCMC &MCMC\\ \hline
			Parameter &{\em Bias} & & & & & & &\\ \cline{2-6}\cline{8-9} 
			$\mu$ & 0.0095 & {\bf 0.0018} &  0.0019 & 0.0083 & 0.0019 & &0.0011 &0.0018\\ 
			$\sigma^2$ & 0.0398 & 0.0053 &  {\bf 0.0022} & 0.0438 & 0.0386  & &0.0018 &0.0053\\ 
			$\alpha$ & -0.1834 & 0.0118 & 0.0337  & -0.1796 &  {\bf 0.0106} & &-0.0091 & 0.0118\\
			$\beta$ & -0.0934 & {\bf -0.0047} & 0.0122  & -0.0904 & -0.0120 & &-0.0049 &-0.0047\\
			$\tau$ & 0.0190 & 0.0054 & 0.0142  & 0.0135 &  {\bf 0.0041} & &-0.0001 &-0.0019\\
			&{\em RMSE} & & & & & & &\\ \cline{2-6}\cline{8-9}
			$\mu$ & 0.0438 & 0.0449 & 0.0448  & 0.0437 & {\bf 0.0426} & & 0.0379 & 0.0449\\
			$\sigma^2$ & 0.0760 & {\bf 0.0617} &  {\bf 0.0305}  & 0.0784 & 0.0752 & & 0.0586 & 0.0617\\
			$\alpha$ & 0.4559 & {\bf 0.4109} & 0.4283  & 0.4537 & 0.4227 & & 0.3886 & 0.4109\\
			$\beta$ & 0.2049 & {\bf 0.1824} & 0.1896  & 0.2030 & 0.1874 & & 0.1737 & 0.1824\\
			$\tau$ & 0.0298 & {\bf 0.0217} & 0.0270  & 0.0283 & 0.0228 & & 0.0180 & 0.0211\\ \hline\hline
		\end{tabular}
	\end{center}
	The bias and RMSE values of different point estimators computed over the $S=500$ simulation replicates, with the lowest values
	in bold. Results on the left are where the misspecified copula is fit using the posterior mean from the conventional (i.e. uncut) and cut (type 1) posteriors, computed exactly using MCMC or approximately using VI. IFM is included for comparison. Results on the right are where the correct copula is fit using the 
	posterior mean from the conventional (i.e. uncut) and cut (type 1) posteriors
	computed exactly using MCMC.
\end{table}

\begin{table}[htbp]
	\caption{Coverage and Copula Model Accuracy in Simulation 1 ($n=500$)} \label{tab:sim1covklN500}
	\begin{center}
		\renewcommand\arraystretch{1.15}
		\begin{tabular}{ccccccccc}
			\hline\hline
			& \multicolumn{5}{c}{Misspecified Copula Fit} & &\multicolumn{2}{c}{Correct Copula Fit} \\ \cline{2-6}\cline{8-9} 
			& Uncut/ & Cut/ & IFM & Uncut/ & Cut/ & &Uncut/ &Cut/\\
			& MCMC & MCMC & &VI &VI & &MCMC &MCMC\\ \hline
			Parameter &\multicolumn{3}{l}{\em Coverage Probabilities} & & & & &\\ \cline{2-6}\cline{8-9} 
			$\mu$ & 0.9580 & {\bf 0.9540} & - & 0.9680 & 0.9600 & &0.9480 &0.9540\\ 
			$\sigma^2$ & 0.8980 & {\bf 0.9640} & - & 0.8880 & 0.9760 & &0.9620 & 0.9640\\
			$\alpha$ & 0.9100 & {\bf 0.9520} & - & 0.9900 & 0.9860 & &0.9520 &0.9520\\
			$\beta$ & 0.9000 & {\bf 0.9380} & - & 0.9900 & 0.9760 & &0.9520 &0.9380\\
			$\tau$ & 0.6220 & {\bf 0.8840} & - & 0.8500 & 0.8560 & &0.9440 &0.9080\\ 
			Component &\multicolumn{3}{l}{\em Mean Predictive  KL Divergence} & & & & &\\ \cline{2-6}\cline{8-9} 
			$f_1$ & 0.0023 & 0.0019 &  {\bf 0.0012} & 0.0024 & 0.0020 & & 0.0010 & 0.0019\\
			$f_2$ & 0.0023 & {\bf 0.0020} & 0.0021 & 0.0022 & {\bf 0.0020} & & 0.0016 & 0.0020\\ 
			$c$ & 0.1525 & 0.1438 & 0.1460 & 0.1545 & {\bf 0.1431} & & 0.0015 & 0.0022\\ \hline\hline
		\end{tabular}
	\end{center}
	Top: coverage probabilities for 95\% credible intervals for each parameter, with
	those closest to 0.95 in bold. Bottom: the mean predictive KL divergence for the two marginals and the copula density for their estimate, with the lowest values 
	in bold. 
	Results on the left are where the misspecified copula is fit using the conventional (i.e. uncut) and cut (type 1) posteriors, computed exactly using MCMC or approximately using VI. IFM is included for comparison. Results on the right are where the correct copula is fit using the conventional (i.e. uncut) and cut (type 1) posteriors
	computed exactly using MCMC.
\end{table}
\newpage

\section{Additional results for Simulation 2}
In this appendix we provide additional results for Simulation~2 when the sample size is $n=100$ and $n=500$. The priors for the misspecified marginals are reported in the text. The priors for the
correctly specified marginals are $\mu \sim \text{N}(0,100^2)$, $\sigma^2 \sim \text{Half-Normal}(0,100^2)$, 
$\alpha \sim \text{Half-Cauchy}(0,5)$, and $\beta \sim \text{Half-Cauchy}(0,5)$.

\begin{table}[htbp]
	\caption{Copula Model Estimation Accuracy in Simulation 2 ($n=500$)} 	
	\label{tab:sim2covklN500}
	\begin{center}
		\renewcommand\arraystretch{1.15}
		\begin{tabular}{lcccccccc}
			\hline \hline
			& \multicolumn{4}{c}{Misspecified Copula Fit} & &\multicolumn{2}{c}{Correct Copula Fit} \\ \cline{2-5}\cline{7-8} 
			& Uncut/ & Cut/ & Uncut/ & Cut/ & &Uncut/ &Cut/\\
			& MCMC & MCMC &VI &VI & &MCMC &MCMC\\ \hline
			Marginal $f_1$ & 0.9092 & {\bf 0.7978} & 0.9036 & 0.8113 &  &0.0012 & 0.0011 \\
			Marginal $f_2$ & 0.2155 & {\bf 0.1543} & 0.2106 & 0.1639 & & 0.0019 &0.0019 \\
			Copula & 0.0172 & {\bf 0.0018} & 0.0105 & {\bf 0.0018} & & 0.0016 & 0.0018\\
			\hline\hline
		\end{tabular}
	\end{center}
	Mean predictive KL divergence metrics for the two marginals and the copula density 
	of the bivariate copula model estimate. The lowest values are in bold. Results are given for the type 2 cut
	and conventional (i.e. uncut) posteriors, computed exactly using MCMC or approximately using variational
	inference. 
	Results are given for both the misspecified copula model (left hand side) and the correctly
	specified copula model (right hand side).
\end{table}

\begin{table}[htbp]
	\caption{Copula Model Estimation Accuracy in Simulation 2 ($n=100$)} 	
	\label{tab:sim2covklN100}
	\begin{center}
		\renewcommand\arraystretch{1.15}
		\begin{tabular}{lcccccccc}
			\hline \hline
			& \multicolumn{4}{c}{Misspecified Copula Fit} & &\multicolumn{2}{c}{Correct Copula Fit} \\ \cline{2-5}\cline{7-8} 
			& Uncut/ & Cut/ & Uncut/ & Cut/ & &Uncut/ &Cut/\\
			& MCMC & MCMC &VI &VI & &MCMC &MCMC\\ \hline
			Marginal $f_1$ & 0.9501 & {\bf 0.8169} & 0.9394 & 0.8458 &  &0.0066 & 0.0066 \\
			Marginal $f_2$ & 0.2722 & {\bf 0.1833} & 0.2651 & 0.2018 & & 0.0095 &0.0095 \\
			Copula & 0.0374 & {\bf 0.0094} & 0.0267 & 0.0120 & & 0.0092 & 0.0094\\
			\hline\hline
		\end{tabular}
	\end{center}
	Mean predictive KL divergence metrics for the two marginals and the copula density 
	of the bivariate copula model estimate. The lowest values are in bold. Results are given for the type 2 cut
	and conventional (i.e. uncut) posteriors, computed exactly using MCMC or approximately using variational
	inference. 
	Results are given for both the misspecified copula model (left hand side) and the correctly
	specified copula model (right hand side).
\end{table}

{
	\def\spacingset#1{\renewcommand{\baselinestretch}%
		{#1}\small\normalsize} \spacingset{1}
	\spacingset{1.0}
	\bibliographystyle{apalike}
	\bibliography{refs}

\begin{thebibliography}{}

\bibitem[Alquier et~al., 2023]{alquier2023estimation}
Alquier, P., Ch{\'e}rief-Abdellatif, B.-E., Derumigny, A., and Fermanian, J.-D.
  (2023).
\newblock Estimation of copulas via maximum mean discrepancy.
\newblock {\em Journal of the American Statistical Association},
  118(543):1997--2012.

\bibitem[Bhat and Eluru, 2009]{bhat2009}
Bhat, C.~R. and Eluru, N. (2009).
\newblock A copula-based approach to accommodate residential self-selection
  effects in travel behavior modeling.
\newblock {\em Transportation Research Part B: Methodological}, 43(7):749--765.

\bibitem[Blei et~al., 2017]{blei2017}
Blei, D.~M., Kucukelbir, A., and McAuliffe, J.~D. (2017).
\newblock Variational inference: A review for statisticians.
\newblock {\em Journal of the American statistical Association},
  112(518):859--877.

\bibitem[Bottou, 2010]{bottou10}
Bottou, L. (2010).
\newblock Large-scale machine learning with stochastic gradient descent.
\newblock In Lechevallier, Y. and Saporta, G., editors, {\em Proceedings of the
  19th International Conference on Computational Statistics (COMPSTAT'2010)},
  pages 177--187. Springer.

\bibitem[Carmona and Nicholls, 2022]{carmona+n22}
Carmona, C. and Nicholls, G. (2022).
\newblock Scalable semi-modular inference with variational meta-posteriors.
\newblock arXiv:2204.00296.

\bibitem[Chen and Fan, 2005]{chen2005pseudo}
Chen, X. and Fan, Y. (2005).
\newblock Pseudo-likelihood ratio tests for semiparametric multivariate copula
  model selection.
\newblock {\em Canadian Journal of Statistics}, 33(3):389--414.

\bibitem[Chernozhukov and Hong, 2003]{chernozhukov2003mcmc}
Chernozhukov, V. and Hong, H. (2003).
\newblock An mcmc approach to classical estimation.
\newblock {\em Journal of econometrics}, 115(2):293--346.

\bibitem[Czado, 2019]{czado2019}
Czado, C. (2019).
\newblock Analyzing dependent data with vine copulas.
\newblock {\em Lecture Notes in Statistics, Springer}.

\bibitem[Daniels and Pourahmadi, 2009]{daniels2009}
Daniels, M.~J. and Pourahmadi, M. (2009).
\newblock Modeling covariance matrices via partial autocorrelations.
\newblock {\em Journal of Multivariate Analysis}, 100(10):2352--2363.

\bibitem[Demarta and McNeil, 2005]{demarta2005}
Demarta, S. and McNeil, A.~J. (2005).
\newblock The t copula and related copulas.
\newblock {\em International Statistical Review}, 73(1):111--129.

\bibitem[Frazier and Nott, 2024]{frazier2022cutting}
Frazier, D.~T. and Nott, D.~J. (2024).
\newblock Cutting feedback and modularized analyses in generalized bayesian
  inference.
\newblock {\em Forthcoming: Bayesian Analysis}.

\bibitem[Genest et~al., 2007]{genest2007metaelliptical}
Genest, C., Favre, A.-C., B{\'e}liveau, J., and Jacques, C. (2007).
\newblock Metaelliptical copulas and their use in frequency analysis of
  multivariate hydrological data.
\newblock {\em Water Resources Research}, 43(9).

\bibitem[Genest et~al., 1995]{genest95}
Genest, C., Ghoudi, K., and Rivest, L.-P. (1995).
\newblock A semiparametric estimation procedure of dependence parameters in
  multivariate families of distributions.
\newblock {\em Biometrika}, 82(3):543--552.

\bibitem[Genest and Ne{\v{s}}lehov{\'a}, 2007]{genest2007}
Genest, C. and Ne{\v{s}}lehov{\'a}, J. (2007).
\newblock A primer on copulas for count data.
\newblock {\em ASTIN Bulletin: The Journal of the IAA}, 37(2):475--515.

\bibitem[Grazian and Liseo, 2017]{grazianliseo17}
Grazian, C. and Liseo, B. (2017).
\newblock {Approximate Bayesian Inference in Semiparametric Copula Models}.
\newblock {\em Bayesian Analysis}, 12(4):991 -- 1016.

\bibitem[Hjort et~al., 2010]{hjort2010bayesian}
Hjort, N.~L., Holmes, C., M{\"u}ller, P., and Walker, S.~G. (2010).
\newblock {\em Bayesian nonparametrics}, volume~28.
\newblock Cambridge University Press.

\bibitem[Hoff, 2007]{hoff07}
Hoff, P.~D. (2007).
\newblock Extending the rank likelihood for semiparametric copula estimation.
\newblock {\em The Annals of Applied Statistics}, 1(1):265--283.

\bibitem[Hoff et~al., 2014]{hoff2014information}
Hoff, P.~D., Niu, X., and Wellner, J.~A. (2014).
\newblock Information bounds for gaussian copulas.
\newblock {\em Bernoulli: official journal of the Bernoulli Society for
  Mathematical Statistics and Probability}, 20(2):604.

\bibitem[Jacob et~al., 2017]{jacob+mhr17}
Jacob, P.~E., Murray, L.~M., Holmes, C.~C., and Robert, C.~P. (2017).
\newblock Better together? {S}tatistical learning in models made of modules.
\newblock arXiv:1708.08719.

\bibitem[Jacob et~al., 2020]{jacob2020unbiased}
Jacob, P.~E., O'Leary, J., and Atchad{\'e}, Y.~F. (2020).
\newblock Unbiased {M}arkov chain {M}onte {C}arlo methods with couplings (with
  discussion).
\newblock {\em Journal of the Royal Statistical Society: Series B},
  82(3):543--600.

\bibitem[Joe, 2005]{joe2005}
Joe, H. (2005).
\newblock Asymptotic efficiency of the two-stage estimation method for
  copula-based models.
\newblock {\em Journal of multivariate analysis}, 94(2):401--419.

\bibitem[Joe, 2014]{joe2014dependence}
Joe, H. (2014).
\newblock {\em Dependence modeling with copulas}.
\newblock CRC press.

\bibitem[Joe and Xu, 1996]{joexu1996}
Joe, H. and Xu, J.~J. (1996).
\newblock The estimation method of inference functions for margins for
  multivariate models.
\newblock Technical Report 166, Department of Statistics, University of British
  Columbia.

\bibitem[Kim et~al., 2007]{kim2007comparison}
Kim, G., Silvapulle, M.~J., and Silvapulle, P. (2007).
\newblock Comparison of semiparametric and parametric methods for estimating
  copulas.
\newblock {\em Computational Statistics \& Data Analysis}, 51(6):2836--2850.

\bibitem[Kucukelbir et~al., 2017]{kucukelbir2017automatic}
Kucukelbir, A., Tran, D., Ranganath, R., Gelman, A., and Blei, D.~M. (2017).
\newblock Automatic differentiation variational inference.
\newblock {\em The Journal of Machine Learning Research}, 18(1):430--474.

\bibitem[Liu et~al., 2009]{liu+bb09}
Liu, F., Bayarri, M.~J., and Berger, J.~O. (2009).
\newblock Modularization in {B}ayesian analysis, with emphasis on analysis of
  computer models.
\newblock {\em Bayesian Analysis}, 4(1):119--150.

\bibitem[Liu and Goudie, 2022]{liu+g20}
Liu, Y. and Goudie, R. J.~B. (2022).
\newblock Stochastic approximation cut algorithm for inference in modularized
  {B}ayesian models.
\newblock {\em Statistics and Computing}, 32(7):1--15.

\bibitem[Loaiza-Maya and Smith, 2019]{loaiza2019VBDA}
Loaiza-Maya, R. and Smith, M.~S. (2019).
\newblock Variational bayes estimation of discrete-margined copula models with
  application to time series.
\newblock {\em Journal of Computational and Graphical Statistics},
  28(3):523--539.

\bibitem[Lunn et~al., 2009]{lunn+bsgn09}
Lunn, D., Best, N., Spiegelhalter, D., Graham, G., and Neuenschwander, B.
  (2009).
\newblock Combining {MCMC} with `sequential' {PKPD} modelling.
\newblock {\em Journal of Pharmacokinetics and Pharmacodynamics}, 36:19--38.

\bibitem[Miller, 2021]{miller2021asymptotic}
Miller, J.~W. (2021).
\newblock Asymptotic normality, concentration, and coverage of generalized
  posteriors.
\newblock {\em Journal of Machine Learning Research}, 22(168):1--53.

\bibitem[Min and Czado, 2010]{min2010}
Min, A. and Czado, C. (2010).
\newblock Bayesian inference for multivariate copulas using pair-copula
  constructions.
\newblock {\em Journal of Financial Econometrics}, 8(4):511--546.

\bibitem[Murray et~al., 2013]{murray2013}
Murray, J.~S., Dunson, D.~B., Carin, L., and Lucas, J.~E. (2013).
\newblock Bayesian gaussian copula factor models for mixed data.
\newblock {\em Journal of the American Statistical Association},
  108(502):656--665.

\bibitem[Nagler et~al., 2022]{nagler2022}
Nagler, T., Kr{\"u}ger, D., and Min, A. (2022).
\newblock Stationary vine copula models for multivariate time series.
\newblock {\em Journal of Econometrics}, 227(2):305--324.

\bibitem[Nelsen, 2006]{nelsen06}
Nelsen, R.~B. (2006).
\newblock {\em An Introduction to Copulas}.
\newblock Springer-Verlag, New York, Secaucus, NJ, USA.

\bibitem[Newey and McFadden, 1994]{newey1994large}
Newey, W.~K. and McFadden, D. (1994).
\newblock Large sample estimation and hypothesis testing.
\newblock {\em Handbook of econometrics}, 4:2111--2245.

\bibitem[Nguyen et~al., 2020]{nguyen2020VI}
Nguyen, H., Aus{\'\i}n, M.~C., and Galeano, P. (2020).
\newblock Variational inference for high dimensional structured factor copulas.
\newblock {\em Computational Statistics \& Data Analysis}, 151:107012.

\bibitem[Oakes, 1994]{oakes1994}
Oakes, D. (1994).
\newblock Multivariate survival distributions.
\newblock {\em Journal of Nonparametric Statistics}, 3(3-4):343--354.

\bibitem[Ong et~al., 2018]{ong+ns16}
Ong, V. M.-H., Nott, D.~J., and Smith, M.~S. (2018).
\newblock {Gaussian variational approximation with factor covariance
  structure}.
\newblock {\em Journal of Computational and Graphical Statistics},
  27(3):465--478.

\bibitem[Ormerod and Wand, 2010]{ormerod2010}
Ormerod, J.~T. and Wand, M.~P. (2010).
\newblock Explaining variational approximations.
\newblock {\em The American Statistician}, 64(2):140--153.

\bibitem[Patton, 2006]{patton2006}
Patton, A.~J. (2006).
\newblock Modelling asymmetric exchange rate dependence.
\newblock {\em International economic review}, 47(2):527--556.

\bibitem[Pitt et~al., 2006]{PitChaKoh2006}
Pitt, M., Chan, D., and Kohn, R. (2006).
\newblock Efficient {B}ayesian inference for {G}aussian copula regression
  models.
\newblock {\em Biometrika}, 93:537--554.

\bibitem[Plummer, 2015]{plummer15}
Plummer, M. (2015).
\newblock Cuts in {B}ayesian graphical models.
\newblock {\em Statistics and Computing}, 25:37--43.

\bibitem[Pompe and Jacob, 2021]{pompe+j21}
Pompe, E. and Jacob, P.~E. (2021).
\newblock Asymptotics of cut distributions and robust modular inference using
  posterior bootstrap.
\newblock arXiv:2110.11149.

\bibitem[Ranganath et~al., 2014]{ranganath14}
Ranganath, R., Gerrish, S., and Blei, D. (2014).
\newblock {Black Box Variational Inference}.
\newblock In Kaski, S. and Corander, J., editors, {\em Proceedings of the
  Seventeenth International Conference on Artificial Intelligence and
  Statistics}, volume~33 of {\em Proceedings of Machine Learning Research},
  pages 814--822, Reykjavik, Iceland. PMLR.

\bibitem[Shen and Wasserman, 2001]{shen2001rates}
Shen, X. and Wasserman, L. (2001).
\newblock Rates of convergence of posterior distributions.
\newblock {\em Annals of Statistics}, pages 687--714.

\bibitem[Silva and Lopes, 2008]{silva2008copula}
Silva, R. d.~S. and Lopes, H.~F. (2008).
\newblock Copula, marginal distributions and model selection: a bayesian note.
\newblock {\em Statistics and Computing}, 18:313--320.

\bibitem[Smith et~al., 2010]{smithmin2010}
Smith, M., Min, A., Almeida, C., and Czado, C. (2010).
\newblock Modeling longitudinal data using a pair-copula decomposition of
  serial dependence.
\newblock {\em Journal of the American Statistical Association},
  105(492):1467--1479.

\bibitem[Smith, 2015]{smith2015}
Smith, M.~S. (2015).
\newblock Copula modelling of dependence in multivariate time series.
\newblock {\em International Journal of Forecasting}, 31(3):815--833.

\bibitem[Smith and Khaled, 2012]{smith2012estimation}
Smith, M.~S. and Khaled, M.~A. (2012).
\newblock Estimation of copula models with discrete margins via bayesian data
  augmentation.
\newblock {\em Journal of the American Statistical Association},
  107(497):290--303.

\bibitem[Smith and Klein, 2021]{smithklein2021}
Smith, M.~S. and Klein, N. (2021).
\newblock Bayesian inference for regression copulas.
\newblock {\em Journal of Business \& Economic Statistics}, 39(3):712--728.

\bibitem[Smith and Maneesoonthorn, 2018]{smithman2018}
Smith, M.~S. and Maneesoonthorn, W. (2018).
\newblock Inversion copulas from nonlinear state space models with an
  application to inflation forecasting.
\newblock {\em International Journal of Forecasting}, 34(3):389--407.

\bibitem[Smith and Vahey, 2016]{smithvahey2016}
Smith, M.~S. and Vahey, S.~P. (2016).
\newblock Asymmetric forecast densities for us macroeconomic variables from a
  {G}aussian copula model of cross-sectional and serial dependence.
\newblock {\em Journal of Business \& Economic Statistics}, 34(3):416--434.

\bibitem[Tan and Nott, 2018]{tan2018gaussian}
Tan, L.~S. and Nott, D.~J. (2018).
\newblock Gaussian variational approximation with sparse precision matrices.
\newblock {\em Statistics and Computing}, 28(2):259--275.

\bibitem[Titsias and Lázaro-Gredilla, 2014]{titsias2014doubly}
Titsias, M. and Lázaro-Gredilla, M. (2014).
\newblock Doubly stochastic variational {B}ayes for non-conjugate inference.
\newblock In Xing, E.~P. and Jebara, T., editors, {\em Proceedings of the 31st
  International Conference on Machine Learning}, volume~32 of {\em Proceedings
  of Machine Learning Research}, pages 1971--1979, Bejing, China. PMLR.

\bibitem[Yu et~al., 2023]{yu+ns21}
Yu, X., Nott, D.~J., and Smith, M.~S. (2023).
\newblock Variational inference for cutting feedback in misspecified models.
\newblock {\em Statistical Science}, 38(3):490--509.

\bibitem[Zeiler, 2012]{zeiler12}
Zeiler, M.~D. (2012).
\newblock {ADADELTA: An adaptive learning rate method. arXiv1212.5701}.

\end{thebibliography}
}

\end{document}